\documentclass[a4paper, 11pt]{article}

\usepackage{amssymb,amsmath,amsthm,longtable,verbatim}

\usepackage{graphicx}

\usepackage[margin=3cm]{geometry}

\numberwithin{equation}{section}

\newcommand{\fsign}{\zeta}

\newcommand{\C}{{\mathbb C}}

\newcommand{\Z}{{\mathbb Z}}\newcommand{\E}{{\mathbb E}}

\newcommand{\crs}{\beta}

\newcommand{\dob}{\omega}
\newcommand{\hcap}{\mathrm{hcap}}
\renewcommand{\H}{{\mathbb H}}

\renewcommand{\P}{{\mathbb P}}

\newcommand{\R}{{\mathbb R}}

\newcommand{\const}{\mathrm{const}\, }
\newcommand{\comp}{\mathcal{C}}

\newcommand{\mesh}{\delta}
\newcommand{\edges}{\text{Edges}}
\newcommand{\vertices}{\text{Vertices}}
\newcommand{\Od}{\Omega^\mesh}

\newcommand\Conf{\mathrm{Conf}}

\newcommand\wind{\mathrm{w}}
\newcommand\pa{\partial}
\newcommand\arb{1}

\renewcommand\Re{\mathrm{Re}\,}
\renewcommand\Im{\mathrm{Im}\,}

\def\res{\mathop{\mathrm{res}}}

\newcommand\dist{\mathrm{dist}}
\newcommand\nor{\vec}

\newcommand\free{\{\text{free}\}}
\newcommand\faces{{\text{Faces}}}
\newcommand\plus{\{{+}\}}
\newcommand\minus{\{-\}}
\newcommand\norF{\tilde{F}}
\newcommand\crossing{\mathcal{I}}
\newcommand\bcond{\mathcal{B}}

\theoremstyle{remark}

\theoremstyle{plain}
\newtheorem{theorem}{Theorem}[section]

\newtheorem{lemma}[theorem]{Lemma} 
\newtheorem{corollary}[theorem]{Corollary} 
\newtheorem{proposition}[theorem]{Proposition} 

\theoremstyle{remark}
\newtheorem{rem}[theorem]{Remark}

\date{}
\title{Smirnov's observable for free boundary conditions, interfaces and crossing probabilities}
\author{Konstantin Izyurov}
 
\begin{document}
\maketitle
 \begin{abstract}
We prove convergence results for variants of Smirnov's fermionic observable in the critical planar Ising model in presence of free boundary conditions. One application of our analysis is a simple proof of a theorem by Hongler and Kyt\"ol\"a on convergence of critical Ising interfaces with plus-minus-free boundary conditions to dipolar SLE(3), and a generalization of this result to an arbitrary number of arcs carrying plus, minus or free boundary conditions. Another application is a computation of scaling limits of crossing probabilities in the critical FK-Ising model with arbitrary number of alternating wired/free boundary arcs. We also deduce a new crossing formula for the spin Ising model.
\end{abstract}
The Stochastic Loewner evolution, introduced by Schramm in \cite{Schramm00}, is a powerful tool in the study of lattice models in two-dimensional statistical mechanics at criticality. In this approach, one describes random geometric shapes arising in the models by planar growth processes. By means of Loewner's equation for the evolution of conformal maps, such processes can be encoded by continuous, real-valued ``driving functions'' (see, e. g., \cite{Lawler-book}).

Schramm's original idea (often called ``Schramm's principle'') was that for certain boundary conditions, natural conformal invariance and ``domain Markov property'' assumptions on a random curve can be restated in terms of its  driving function. In the scaling limit, these properties identify the latter as a Brownian motion $B_{\kappa t}$, where the intensity $\kappa>0$ represents the universality class of the model. This approach, pursued in particular in \cite{Smirnov01, LSW, SS_harm, SS_DGFF, Camia_Newmann, CHS2, CHS3}, was extremely fruitful. In a more general setup (e. g. for more complicated boundary conditions), the driving processes are typically described by Brownian motion $B_{\kappa t}$ with time-dependent drifts; these processes do not admit such a simple axiomatic characterization anymore, and a lot of work has been done  (see e. g. \cite{LSW_restriction, BB, BBK05, LK, Dub07, Zhan-LERW, Dub_FF, IK, FK, KP}) in order to understand them both in general and in relation to concrete lattice 
models.

One celebrated result in the area is the proof of conformal invariance of fermionic observables in the critical Ising model \cite{Smirnov06, CHS2}, leading in particular to the proof that interfaces in the model and its random cluster representation converge to SLE${}_3$ and SLE${}_{\frac{16}{3}}$ respectively \cite{CHS3}. This result was extended in \cite{Izy} to radial and multiple SLE and to multiply-connected domains with suitable analogs of Dobrushin boundary conditions. Another very interesting case, namely that of free boundary conditions, was treated by Hongler and Kyt\"ol\"a \cite{HonKyt}, who proved a conjecture of \cite{BB} that interfaces in the critical Ising model on a simply-connected domain with plus-minus-free boundary conditions converge to the dipolar SLE${}_3$, i. e., the SLE$_{\kappa}(\rho)$ process \cite{LSW_restriction, SW} with $\kappa=3$ and $\rho=-\frac32$. The beautiful proof of Hongler and Kyt\"ol\"a was quite complicated, the main source of technical difficulties 
being that they did not 
use discrete analytic or discrete harmonic functions directly. Instead, the key ingredient of the proof, namely the  computation of the scaling limit of a martingale observable, was obtained using the convergence of FK-Ising interfaces to SLE$_\frac{16}{3}$. The argument relied on explicit formulae for certain SLE martingales, meaning 
 that it was hardly amenable to generalizations.

In the present paper, we point out another approach, employing a direct generalization of Smirnov's observable \cite{CHS2} to the case when a free boundary arc is present. Our observable is a discrete holomorphic function $F(z)$ defined on a discrete domain $\Od$ whose boundary is divided into three arcs. For any $z$, it possesses a martingale property with respect to the interface in the critical Ising model on $\Od$ with $+,-$, and free boundary conditions on these arcs. Our main new observation is that it is possible to treat the arising boundary value problem, in particular, identify the boundary conditions on the free arc. This enables us to prove the convergence of this observable to a conformally covariant scaling limit using techniques developed in \cite{CHS2}. 

In various degenerate cases, our observable coincides with previously known ones. When there is no free boundary arc, one gets the original Smirnov's observable \cite{Smirnov10, CHS2}. When there is no ``$-$'' arc, one gets Smirnov's observable for FK-Ising model \cite{Smirnov10}, which was used to prove convergence of FK-Ising interfaces to SLE$_{\frac{16}{3}}$. Finally, when $z$ is on the boundary, one recovers the observable originally employed in \cite{HonKyt}. 

An advantage of our proof is that it readily generalizes to more complicated geometries. We stick for simplicity to the case of simply connected domains with arbitrary number of boundary arcs, carrying $+$, $-$, or free boundary conditions, although in principle, using techniques employed in \cite{Izy} for fixed boundary conditions, one can extend these results to multiply connected domains. Our result (Theorem \ref{thm: interfaces}, see Section \ref{sec: interfaces} for details and notation) states that in the limit, the interfaces are conformally invariant and are described by the chordal Loewner evolution with a driving force $a_1(t)$ satisfying the SDE
$$
da_1(t)=\sqrt{3}dB_t+D(\{a_i(t),b_i(t)\})dt,
$$
where $a_i(t)$ and $b_i(t)$ are the images under the Loewner map at time $t$ of the points where boundary conditions change (from $+$ to $-$ and from fixed to free, respectively) and $D$ is a quadratic irrational drift function. We provide a general, explicit formula for $D$; for example, in the cases of four and five marked points with $+/-/+/$ free and $+/-/\text{free}/+/\text{free}$ boundary conditions one has 
\begin{gather*}
D_{+/-/+/\text{free}}=-\frac{3}{a_1-a_2}-\frac{3/2}{a_1-b_1}+\frac{3}{a_2+a_1-2b_1},\\
D_{+/-/\text{free}/+/\text{free}}=-\frac{3/2}{a_1-b_1}-\frac{3/2}{a_1-b_2}-\frac{3/2}{a_1-b_3}+3\left(a_1-b_3-\sqrt{(b_3-b_2)(b_3-b_1)}\right)^{-1},
\end{gather*}
where we took $b_2=\infty$ and $b_4=\infty$ respectively.

As a byproduct, we get the convergence of crossing probabilities in the random cluster representation of the critical Ising model. Given a simply-connected domain with $2k$ marked boundary arcs, with alternating wired/free boundary conditions on these arcs, pick any subset of wired arcs and consider the probability that these arcs belong to the same cluster. We prove that the scaling limits of those quantities exist and are conformally invariant, and after mapping to the half-plane, are expressed by explicit quadratic irrational functions. This part extends the case of four boundary arcs treated in \cite{CHS2}. Our technical novelty, which seems essential in the general case, is to consider the observables in the spin representation and then transfer the results by the Edwards-Sokal coupling. To the best of our knowledge, our formulae did not appear explicitly neither in the mathematics nor in the physics literature before.

Incidentally, the explicit computation of the drift functions $D$ and of the crossing probabilities involves the same interpolation problem as the one encountered in \cite{CHHI} when computing spin correlations in the Ising model. The computation as in Proposition~\ref{prop: linear} has been used to make the result of \cite{CHHI} completely explicit, matching the physics literature predictions. A similar computation also allows one to extend the results of \cite{CHHI} to more general boundary conditions (involving free arcs), and to obtain new correlation formulas, which is a subject of a forthcoming paper. 

The paper is organized as follows. In Section \ref{sec: observable}, we introduce the observable and prove its convergence in the scaling limit.  In Section \ref{sec: multipoint}, we work out the case of an arbitrary number of boundary arcs, and obtain as a corollary the proof of conformal invariance of the FK-Ising crossing probabilities. In Section \ref{sec: interfaces}, we derive the convergence of the interfaces. This part is quite standard and employs the same argument as e. g. in \cite{Izy, Zhan-LERW}, eventually going back to \cite{LSW}. In Section \ref{sec: corollaries}, we give the explicit formulae for the observables and hence for the drift terms. We also discuss a new spin crossing formula for the Ising model in topological  rectangles. In Appendix, we collect for convenience of the reader proofs of several facts (most of which are well known) used in the paper.

\vskip 0.5cm
\noindent \textbf{Acknowledgements.} The author is grateful to Dmitry Chelkak, Antti Kemppainen and Kalle Kyt\"ol\"a for valuable conversations, and to the referee for a number of suggestions on improving the text of the paper. Work supported by Academy of Finland. 

\section{Convergence of the observable}
\label{sec: observable}
Let $\Od$ be a simply connected discrete domain of mesh size $\delta>0$, i. e. a subset of faces of $\delta\Z^2$ that forms a connected and simply connected polygonal domain, together with all the vertices and edges incident to those faces. We assume that the boundary of $\Od$ is divided into three arcs, denoted by $\free$, $\plus$ and $\minus$ and separated by vertices $a,b,c$; this subdivision will specify the boundary conditions $\bcond^\delta$ in our model.

The Ising model on $\Od$ with the boundary conditions $\bcond^\delta$ is a random assignment \mbox{$\sigma:\faces(\Od)\to {\pm 1}$} of spins to the faces of $\Od$, with the probability measure given by
\begin{equation}
\label{eq: defIsing}
\P(\sigma)= Z^{-1}\exp[\beta\sum\limits_{u\sim u'}\sigma(u)\sigma(u')],\quad Z:=\sum\limits_\sigma\exp[\beta\sum\limits_{u\sim u'}\sigma(u)\sigma(u')]. 
\end{equation}
The sum in the exponential in (\ref{eq: defIsing}) is then taken over the set of pairs of adjacent faces of $\delta\Z^2$ separated by an edge of $\Od$, except for those edges that belong to the $\free$ arc. The spins on the faces of $\delta\Z^2\backslash \Od$ adjacent to $\plus$ and $\minus$ arcs are assumed to be non-random and equal to $+1$ and $-1$, respectively. 

Let us rewrite the definition (\ref{eq: defIsing}) in the low-temperature expansion. Denote by $\Conf(\Od,z_1,z_2)$, where $z_{1,2}\in \vertices(\Od)$, the set of all $S\subset \edges(\Od)$ such that all vertices of $\Od$ have an even degree in $S$, except for $z_{1,2}$, which have an odd degree. Assign for convenience the value $+1$ to the faces of $\delta\Z^2\backslash \Od$ adjacent to $\free$ (we could as well take $-1$). Given a spin configuration $\sigma:\faces(\Od)\to {\pm 1}$, draw the edges separating faces with different spins (see Fig. \ref{Fig: domain}); this gives a bijection from $\{\pm 1\}^{\faces(\Od)}$ to $\Conf(\Od,a,b)$, thus endowing the latter set with the probability measure
\begin{equation}
\label{eq: low-temperature}
\P(S)= Z^{-1}x^{|S\backslash \free|}, \quad Z=Z(\Od,a,b,\free):=\sum\limits_{S\in\Conf(\Od,a,b)}x^{|S\backslash \{\text{free}\}|}, 
\end{equation}
where $x=\exp[-2\beta]$ will be, from now on, set to its critical value, 
$$
x=x_c=\sqrt{2}-1.
$$

\begin{figure}[t]
\begin{center}
\includegraphics[width=0.7\textwidth]{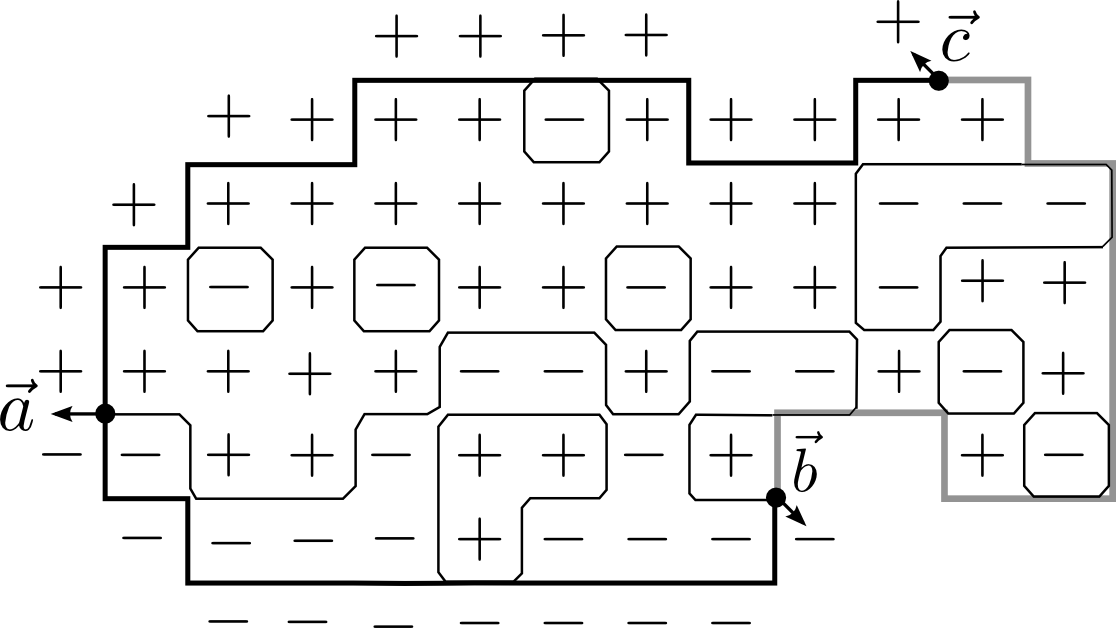}
\end{center}
\caption{\label{Fig: domain} An example of a spin configuration and the corresponding configuration $S\in\Conf(\Od,a,b)$. The free boundary arc $(bc)$ is in gray; $S$ is drawn as if the spins adjacent to $(bc)$ were $+1$. }
\end{figure}

It is convenient to endow $\Od$ with an additional decoration. First, we add a vertex to the midpoint of each edge; it is clear that the Ising model on the original graph $\Od$ in the low-temperature expansion is equivalent to one on the new graph with the weights $\sqrt{x}$ per half-edge. Second, we add vertices at \emph{corners} of faces, i. e. for each vertex $v$ of $\Od$, we add four vertices $c_j$ at $v+\frac{\sqrt{2}\delta}{4} e^{\frac{i\pi}{4}+\frac{i\pi}{2}j}$, $j=0,1,2,3$, and connect each $c_j$  by an edge to $v$; here is what we obtain: 
\begin{center}
\includegraphics[width=0.5\textwidth]{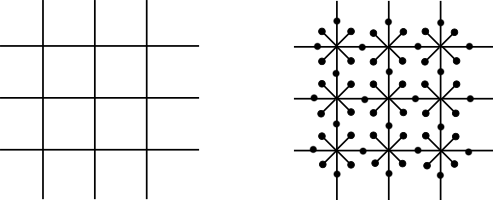}
\end{center}
We will often identify an edge of $\Od$ with its midpoint, and a corner edge with its endpoint. We assign the weight $\sqrt{x}\cos\frac\pi8$ to each corner edge, and extend the definition of $\Conf(\Od,z_1,z_2)$ and of the corresponding partition functions by allowing $z_{1,2}$ to be corners or midpoints of edges. We then use the convention that $x^{|S\backslash \free|}$ is understood as the product of the weights of edges in the \emph{decorated} graph that constitute $S\backslash \free$.

By a \emph{discrete outer normal} at a vertex $v$ on the boundary of $\Od$ we mean an edge of the decorated graph which connects $v$ to a corner or to a midpoint of an edge not in $\Od$, oriented from $v$ to the latter. An oriented edge $e$ can be viewed as a complex number, and we introduce another number $\eta_e\in\C$ and the straight line $l_e\subset \C$ by 
$$
\eta_e:= \left(\frac{ie}{|e|}\right)^{-\frac12},\quad l_e:=\eta_e\R;
$$
note that $\eta_e$ is defined up to sign which we will specify when necessary.
 
The definition of the \emph{fermionic observable} is a natural generalization of the one proposed by Smirnov \cite{Smirnov10, CHS2}. Let $\nor{a}$ be a discrete outer normal \emph{edge} at $a$, and denote by $\nor{b},\nor{c}$ the discrete outer normal \emph{corners} at $b$ and $c$, adjacent to the $\minus$ and to the $\plus$ boundary arcs respectively, see Fig \ref{Fig: domain}. Suppose that $z\neq \nor{a}$ is either a corner of $\Od$, or a discrete outer normal, or an edge midpoint of $\Od\backslash \free$ (below we give a separate definition for $z\in\free$). We put
\begin{equation}
\label{eq: obs}
F(z)=F_{\Od,\bcond^\delta}(z)=i\eta_{\nor{a}}\sum\limits_{S\in \Conf(\Od,\nor{a},z)}x^{|S\backslash \free|}e^{-i\wind(S)/2}
\end{equation}
 the winding factor $\wind(S)$ being defined as follows. Decompose $S$ into a union of loops and a path $\gamma$ starting with the edge $\nor{a}$ and ending at $z$ in such a way that no edge is traced twice, and that the loops and the path $\gamma$ do not cross each other or themselves transversally; we do allow them to have self-touchings or mutual touchings at vertices. Then $\wind(S)$ is defined to be the winding number of the tangent vector of $\gamma$ around zero, that is, the net angle by which this vector turns on the way form $\nor{a}$ to $z$. It is easy to see that the winding factor $e^{-i\wind(S)/2}$ does not depend on the decomposition of $S$, thus the observable is well defined.

 Define, for $b>0$ and $z\in\H=\{z\in\C:\Im z>0\}$, 
\begin{equation}
f_{\H, b}(z):=\frac{z-2b}{\sqrt{\pi} z\sqrt{b-z}}. 
\label{eq: obs_hp}
\end{equation}
Given a simply connected domain $\Omega$ with boundary conditions $\bcond$ (specified by marked points $a,b,c$ such that $c\neq a,b$), let $\varphi_{\bcond}$ denote a conformal map from $\Omega$ to $\H$ such that $\varphi_\bcond(a)=0$  and $\varphi_\bcond(c)=\infty$, and denote
$$
f_{\Omega, \bcond}(z)=(\varphi_\bcond'(z))^{\frac12}f_{\H,\varphi_{\bcond}(b)}(\varphi_{\bcond}(z)).
$$

Throughout this paper, we say that discrete domains $(\Od,\bcond^\delta)$ \emph{approximate} $(\Omega,\bcond)$ if $\Od$ converges to $\Omega$ as $\delta$ tends to zero in the sense of Carath\'eodory, and the boundary points $a^\delta,b^\delta,c^\delta,\ldots$ specifying $\bcond^\delta$ converge as prime ends to their counterparts in $\bcond$. 
 \begin{proposition}
 \label{prop: convergence}
Suppose the domains $(\Od,\bcond^\delta)$ approximate a simply connected domain $(\Omega,\bcond)$, and suppose that $\bcond$ are such that $c\neq a,b$. Then 
$$
\norF:=\frac{F_{\Od,\bcond^\delta}(\cdot)}{\sqrt[4]{2}\sqrt{\delta}Z(\Od,\nor{a}^\delta,\nor{b}^\delta,\free)}\longrightarrow f_{\Omega, \bcond}(\cdot)
$$ 
uniformly on compact subsets of $\Omega$ (here $F$ is viewed as a function on edges of $\Od$).
\end{proposition}
 
\begin{rem}
(i) This is a ``bulk'' version of \cite[Theorem 10]{HonKyt}, which was the main ingredient in the proof of the convergence of interfaces.

\noindent (ii) We do not assume that $a\neq b$; in particular, we do allow $a^\delta=b^\delta$, i. e. no $-$ arc. This could be used to prove the convergence of interfaces beyond the time it hits the free arc, the result obtained in a different way in \cite{BDH}.
\end{rem}

\noindent (iii) The proof goes along the lines of \cite{CHS2}; we prove the discrete holomorphicity of $F$ and establish Riemann type boundary conditions, and then we consider the imaginary part of the discrete integral of $F^2$, transforming these boundary conditions to Dirichlet ones. The only additional work is to take into account the free part of the boundary. 
 
\begin{lemma}
\label{lemma: dhol}
The function $F_{\Od,\bcond^\delta}(\cdot)$ satisfies the \emph{s-holomorphicity condition} \cite{CHS2}
\begin{equation}
 \text{Proj}_{l_q}F(z)= F(q)
\label{eq: s-hol}
\end{equation}
whenever a corner $q$ and an edge $z$ are adjacent, $z\notin \free$ and $z\neq \nor{a}$, where $\text{Proj}_{l}w$ denotes the orthogonal projection of the complex number $w$ to the line $l\subset\C$.
\end{lemma}
\begin{proof} The proof is the same as for the observable in \cite{CHS2}; see Appendix. \end{proof}

 The equation (\ref{eq: s-hol}) allows one to recover the value of $F(z)$ from the values at any two corners adjacent to $z$. In particular, if $z$ is oriented, and if $q_{1,2}$ are the corners adjacent respectively to the beginning and to the end of $z$ \emph{on its left side}, then 
 \begin{equation}
 \label{eq: continuation}
 F(z)=\sqrt{2}e^{-\frac{i\pi}{4}}F(q_2) + \sqrt{2}e^{\frac{i\pi}{4}}F(q_1).
 \end{equation}
When an edge $z$ belongs to $\free$, the s-holomorphicity is incompatible with the definition~(\ref{eq: obs}), and it is more convenient to use the former rather then the latter. Thus, for $z\in \free$, we define $F(z)$ by (\ref{eq: continuation}), orienting $z$ to have the domain on its left.

 \begin{lemma}
\label{lemma: bc}
\begin{enumerate}
 \item If $z\neq \nor{a}$ is a discrete outer normal, then 
 \begin{equation}
 \label{eq: bc_cb}
 F_{\Od,\bcond^\delta}(z)=\eta_{z}Z(\Od,\nor{a},z,\free),
 \end{equation}
 where the sign of $\eta_{z}$ is defined by extending continuously from $\eta_{\nor{a}}$ (used in the definition of $F$) along the boundary in the counterclockwise direction. 
 \item If an edge $z\in\free$ is oriented to have $\Od$ on its left, then 
  \begin{equation}
 \label{eq: bc_bc}
 F_{\Od,\bcond^\delta}(z)\in iz^{-\frac{1}{2}}\R
 \end{equation}
\end{enumerate}
\end{lemma}
\begin{proof}
The first statement follows from an easy observation that if $z$ is a discrete outer normal, then for any $S\in\Conf(\Od,\nor{a},z)$ one has $i\eta_{\nor{a}} e^{-i\wind(S)/2}=\eta_z$, where $\eta_z$ is defined as in the assertion. For the second statement, note that if $q_{1,2}$ are the two corners inside $\Od$ adjacent to the beginning and to the end of $z$, respectively, then 
\begin{equation}
F(q_2)= e^{-\frac{\pi}{4}i}F(q_1).
\label{eq: bc_proof}
\end{equation}
To prove this identity, use the bijection $p: S\mapsto S \bigtriangleup (q_1\cup z\cup q_2)$ from $\Conf(\Od,\nor{a},q_1)$ to $\Conf(\Od,\nor{a},q_2)$, where $\bigtriangleup$ stands for the symmetric difference. For any $S\in \Conf(\Od,a,q_1)$, one has $\wind(p(S))=\wind(S)+\frac{\pi}{2}$ and $|S\backslash \free|=|p(S)\backslash \free|$, hence (\ref{eq: bc_proof}) follows. It is clear that (\ref{eq: bc_proof}) and the definition (\ref{eq: continuation}) of $F(z)$ imply that $F(z)\in e^{-\frac{\pi}{8}i}l_{q_1}=e^{\frac{\pi}{8}i}l_{q_2}=iz^{-\frac{1}{2}}\R$.
\end{proof}

The boundary conditions established in Lemma \ref{lemma: bc} may be informally summarized as follows: if $n_z$ denotes the outer normal vector at $z$, then (\ref{eq: bc_cb}) implies that $F(z)(in_z)^{\frac{1}{2}}\in \R$ when $z\in\plus\cup\minus$, and (\ref{eq: bc_bc}) means that $F(z)(in_z)^{\frac{1}{2}}\in i\R$ when $z\in \free$. Of course (\ref{eq: bc_cb}) also holds for the discrete outer normals at the vertices of the free boundary arc; however, the s-holomorphicity (\ref{eq: s-hol}) fails for $z\in\free$ and $q$ an adjacent discrete outer normal. Therefore, as we will see, the relevant boundary conditions for $\free$ are given by (\ref{eq: bc_bc}) and not by (\ref{eq: bc_cb}).

Define two functions $H^\circ$, $H^\bullet$ on the faces and on the vertices of $\Od$ respectively by the following rule. Set $H^\circ(u_0)=0$, where $u_0$ is some face of $\delta\Z^2\setminus \Od$ adjacent to $\plus\cup\minus$; we will show in a moment that the choice of $u_0$ is not important. Next, if $v$ is a vertex of $\Od$, $u$ is a face adjacent to $v$, and $q$ is the corner adjacent to both of them, then we put
\begin{equation}
  H^\bullet(v)-H^\circ(u)=\sqrt{2}\delta |\norF(q)|^2.
  \label{eq: defH}
\end{equation}
Note that $H$ is well defined by this rule at all the vertices and all the faces of $\Od$ and at the faces adjacent to $\plus\cup\minus$. Indeed, if $q_{1,2,3,4}$ are the corners (in cyclic counterclockwise order) adjacent to an edge $e$ of $\Od\backslash \free$, then $l_{q_1}$ and $l_{q_2}$ are orthogonal to $l_{q_3}$ and $l_{q_4}$, respectively. Hence, the s-holomorphicity condition (\ref{eq: s-hol}) ensures that
\begin{equation}
\label{eq: orthogonal}                                                                                                                                                                                                                                             
|{F}(q_1)|^2+|{F}(q_3)|^2=|{F}(e)|^2=|{F}(q_2)|^2+|{F}(q_4)|^2,                                                                                                                                                                                                                                           \end{equation}
which means that summing (\ref{eq: defH}) around this edge yields zero. We also mention that $H$ is a version of the imaginary part of the discrete integral of $F^2$, namely, if $v\sim v'$ (respectively, $u\sim u'$), then
\begin{eqnarray}
H^\bullet(v)-H^\bullet(v')= \Im \left[\norF^2\left(\frac{v+v'}{2}\right)(v-v')\right],\label{eq: intHbullet}\\
H^\circ(u)-H^\circ(u')= \Im \left[\norF^2\left(\frac{u+u'}{2}\right)(u-u')\right].\label{eq: intHcirc}
\end{eqnarray}
The first identity is easily checked by expressing both sides in terms of $\norF(q_{1,2})$, where $q_{1,2}$ are two corners adjacent to the edge $(vv')$ on the same side, using (\ref{eq: defH}) and (\ref{eq: continuation}) respectively. The second one is similar. 
\begin{lemma}
\label{lemma: H}
 The functions $H^{\bullet,\circ}$ satisfy the following properties:
 \begin{itemize}
 \item $H^\circ\equiv 0$ at the faces of $\delta\Z^2\setminus \Od$ adjacent to $\plus\cup\minus$, and $H^\bullet\equiv 1$ at the vertices of $\free$. For any vertex $v\in\plus\cup\minus$ one has $H^{\bullet}(v)\geq 0$, and for any face $u\in \Od$ adjacent to $\free$ one has $H^\circ(u)\leq 1$.
  \item The inequalities $\Delta H^\bullet(v)\geq 0$,  $\Delta H^\circ(u)\leq 0$, where $\Delta$ stands for the standard discrete Laplacian, hold true for any \emph{interior} vertex $v\in \Od$ and any face $u\in\Od$ not adjacent to $\free$. Moreover, they also hold for the boundary vertices $v\in\plus\cup\minus$, $v\neq a$ and for the faces $u\in \Od$ adjacent to $\free$, with the Laplacian modified on the boundary:  $\Delta H(z) =\sum\limits_{w\sim z}c(z,w)(H(w)-H(z))$, where $c(z,w)=1$ unless $w$ is either a face of $\mesh\Z^2\setminus \Od$ adjacent to $\free$, or a vertex of $\mesh\Z^2\setminus \Od$ adjacent to $\plus\cup\minus$, in which case $c(z,w):=2(\sqrt{2}-1)$ and we set $H^\circ(w)=1$ (respectively, $H^\bullet(w)=0$).
  \end{itemize}
\end{lemma}
\begin{proof}
 Suppose $v$ is a vertex of $\plus\cup\minus$, $u_1\sim u_2$ are two faces of $\delta\Z^2\setminus \Od$ adjacent to $v$ and $q_{1,2}=(v+u_{1,2})/2$ are the corresponding corners. Then, (\ref{eq: bc_cb}) implies that $|F(q_1)|=|F(q_2)|$, hence $H^\circ(u_1)=H^\circ(u_2)$. Similarly, if $v_1\sim v_2$ are two vertices of $\free$, then (\ref{eq: bc_proof}) implies that $H^\bullet(v_1)=H^\bullet(v_2)$. Consequently, $H^\circ\equiv 0$ along $\plus\cup\minus$, and $H^\bullet$ is a constant along $\free$. However, we have $H^\bullet(b)=\sqrt{2}\delta |\norF(\nor{b})|^2=1$ by (\ref{eq: bc_cb}) and the normalization of $\norF$; thus $H^\bullet\equiv 1$ on $\free$. The inequalities follow readily from (\ref{eq: defH}).
 
 For the second clause of the lemma, see \cite[Proposition 3.6]{CHS2} or Appendix.
\end{proof}

Given boundary conditions $\bcond^\delta$, denote by $\bcond^\delta_{1}$ the boundary conditions specified by the marked points $(a^\delta_1,b^\delta_1,c^\delta_1)=(b^\delta,b^\delta,c^\delta)$, i. e. with the same $\free$ arc as $\bcond^\delta$ but with no $\minus$ arc.

\begin{lemma}
\label{lemma: aa}
 As $(\Od,a^\delta,c^\delta)$ approximates $(\Omega,a,c)$, the function $\norF_{\Od,\bcond_{\arb}^\delta}(\cdot)$ converges uniformly on compact subsets of $\Omega$ to $f_{\Omega,\bcond_{\arb}}=(\varphi'_{\bcond}(\cdot))^\frac12(-\pi\varphi_{\bcond}(\cdot))^{-\frac{1}{2}}$.
\end{lemma}
In fact, this lemma is already contained in \cite{Smirnov06} since $F_{\Od,\bcond^\delta_{\arb}}(\cdot)$ is nothing but Smirnov's FK-Ising observable for the medial lattice. This can be seen either by observing that it solves the same discrete boundary value problem, or by the Edwards-Sokal coupling.
\begin{proof}
  Define $H^{\circ,\bullet}=H^{\circ,\bullet}_{\bcond_{\arb}^\delta}$ by (\ref{eq: defH}) with $\norF=\norF_{\Od,\bcond_{\arb}^\delta}$. By Lemma \ref{lemma: H}, $0=\min_{\partial \Od}H^\circ=\min_{\Od}H^\circ$ (since $H^\circ$ is superharmonic) and  $\max_{\Od}H^\bullet=\max_{\partial \Od}H^\bullet=1$ (since $H^\bullet$ is subharmonic except from $\free$ where is is equal to 1). Taking into account (\ref{eq: defH}), we infer that $0\leq H^{\circ,\bullet}\leq 1$. By \cite[Theorem 3.12]{CHS2}, the functions $F_{\Od,\bcond^\delta_{\arb}}$ form an equicontinuous family on compact subsets of $\Omega$, and thus have subsequential limits. To prove the lemma, it suffices to show that any such limit $f$ must be equal to $f_{\Omega,\bcond_{\arb}}$
  
  Since $F_{\Od,\bcond^\delta_{\arb}}$ are discrete holomorphic, $f$ has to be a holomorphic function; moreover,  by (\ref{eq: intHbullet}) -- (\ref{eq: intHcirc}), $H^{\circ,\bullet}$ then converge to $h(w)=\Im\int^w f^2(z)dz$. Let us establish the boundary conditions for $h$. By the subharmonicity of $H^\bullet$ and the superharmonicity of $H^\circ$, if $v$ is a vertex, $u$ is a face adjacent to $v$, and if $\gamma\subset \partial\Od$, then 
 $$
\text{hm}^\circ(\gamma,u)\min\limits_{\gamma}H^\circ+(1-\text{hm}^\circ(\gamma,u))\min\limits_{\Od}H^\circ\leq H^\circ(u)\leq $$
\begin{equation}
\label{eq: hm}
H^\bullet (v)\leq \text{hm}^\bullet(\gamma,v)\max\limits_{\gamma}H^\bullet+(1-\text{hm}^\bullet(\gamma,v))\max\limits_{\Od}H^\bullet 
\end{equation}
where $\text{hm}^{\bullet,\circ}(\gamma,\cdot)$  denotes the discrete harmonic measure, i. e. the probability that the simple random walk (to be precise, the one corresponding to the modified Laplacian) on vertices (respectively, on faces) of $\Od$ started from $v$ (respectively, $u$) will exit $\Od$ at $\gamma$. When $\delta\to 0$, both discrete harmonic measures converge to the continuous harmonic measure (see \cite{CHS1}), hence, taking $\gamma = \plus$ (respectively, $\gamma = \free$)  in (\ref{eq: hm}) shows that $h(z)\to 0$ as $z\to \plus$ (respectively, $h(z)\to 1$ as $z\to \free$). Since $h$ is bounded, these boundary conditions determine it uniquely as the harmonic measure of $\free$, or, in terms of the conformal map:
$$
h(z)= 1-\frac{1}{\pi}\Im\log(\varphi_{\bcond}(z)).
$$
Differentiating and taking the square root concludes the proof.
 \end{proof}
 
 \begin{rem}
\label{rem: at_a}
  In fact, two implicit assumptions on $\Od$ were used in the proof: first, that an \emph{edge} outer normal $\nor{a}$ at $a=b$ exists; second, that no discrete outer normal edge has both endpoints in $\Od$. Neither of these minor technicalities is essential. We could take $\nor{a}$ in the definition of the observable to be a copy of $\nor{b}$ slightly turned towards the free arc. If an edge violates the second assumption, one should view it as two distinct normals at different points of $\Od$, with the endpoints declared not to belong to $\Od$, implying the corresponding extension of $H^\bullet=0$ to those endpoints when applicable.
 \end{rem}
 
 \begin{proof}[Proof of Proposition \ref{prop: convergence}.] 
 As in the proof of Lemma \ref{lemma: aa}, we have $0=\min_{\partial \Od}H^\circ=\min_{\Od}H$. This time we have no corresponding upper bound, since the subharmonicity of $H^\bullet$ fails at $a^\delta$, and we have no control on its value there. 
 
 Nevertheless, assume for a moment that for any $r>0$, the functions $H^\bullet$ are uniformly bounded on $\Od\cap \{z:|z-a|>r\}$. Then, arguing as in the proof of Lemma \ref{lemma: aa}, we see that $F_{\Od,\bcond^\delta}$ and $H^{\bullet,\circ}$ have subsequential limits $f$ and $h=\Im\int^w f^2(z)dz$ respectively, and that any subsequential limit $h$ must satisfy the following properties: $h\geq 0$ is a harmonic function, bounded on each $\Omega\cap \{z:|z-a|>r\}$, and such that $h\equiv 1$ on $\free$, $h\equiv 0$ on $\minus\cup\plus$. 
  
 These boundary conditions imply that $h$ is a sum of the Poisson kernel at $a$ with non-negative mass and the harmonic measure of the boundary arc $\free$, that is, 
 $$
  h(z)=1-\frac{1}{\pi}\Im\left[ \log (\varphi_\bcond(z)-\varphi_\bcond(b))]+ \frac{\alpha}{\varphi_\bcond(z)}\right],\quad \alpha\geq 0. 
 $$
 Moreover, it follows from \cite[Remark 6.3]{CHS2} and the first clause of Lemma \ref{lemma: H} that the outer normal derivative of $h$ is non-negative on $\free$, more precisely, that there is no point on $(bc)$ such that $h\geq 1$ in a neighborhood of that point. Observe that the derivative 
 $$
 \partial_w\left(\log (w-\varphi_\bcond(b))]+ \frac{\alpha}{w}\right)=\frac{1}{w-\varphi_\bcond(b)}-\frac{\alpha}{w^2}
 $$
 has two simple zeros on $(\varphi_\bcond(b);\infty)$ if $\alpha>4\varphi_\bcond(b)$ and a simple zero in $\H$ if $0<\alpha<4\varphi_\bcond(b)$. The former is impossible by the normal derivative condition, while the latter would imply that $f=\sqrt{2\partial_z{h}}$ is not a single-valued function in $\Omega$, and thus is also impossible. Hence, $\alpha=0$ or $\alpha=4\varphi_\bcond(b)$.
 
 Let us check that if $a\neq b$, then the first alternative cannot hold. If it did, this would mean by Lemma \ref{lemma: aa} that $\norF_{\bcond^\delta_{\arb}}(\cdot)$ and $\norF_{\bcond^\delta}(\cdot)$ have the same limit (hereinafter we drop $\Od$ from subscripts). This in its turn would imply that if we consider $F^{\dagger}=\norF_{\bcond^\delta}-\norF_{\bcond^\delta_{\arb}}$ and define the corresponding discrete integral $H^\dagger$ by (\ref{eq: defH}), then $H^\dagger$ tends to a constant. Look at the  values of $F^{\dagger}$ at $\nor{b}$ and $\nor{c}$. By Lemma \ref{lemma: bc}, one has $\norF_{\bcond^\delta}(\nor{c})=2^{-\frac14}\mesh^{-\frac12}\eta_{\nor{c}}=\norF_{\bcond^\delta_{\arb}}(\nor{c})$ and $\norF_{\bcond^\delta}(\nor{b})=2^{-\frac14}\mesh^{-\frac12}\eta_{\nor{b}}=-\norF_{\bcond^\delta_{\arb}}(\nor{b})$, where the signs of $\eta$'s are chosen by extending continuously from $\eta_{\nor{a}}$ along the boundary in the counterclockwise direction. Consequently, $H^\dagger\equiv 2$ along $(a^\delta b^\delta)$ and $H^\dagger\equiv 0$ along $(b^\delta a^\delta)$. Hence by (\ref{eq: hm}) it cannot 
tend to a constant, which completes the proof that $\alpha=4\varphi_\bcond(b)$. Thus, the subsequential limit of $\tilde{F}$ is unique and is given by
 $$
 \sqrt{2\partial_z h(z)}=(\varphi_\bcond'(z))^\frac12\left(\frac{\varphi^2_\bcond(z)-4\varphi_\bcond(z)\varphi_\bcond(b)+4\varphi^2_\bcond(b)}{\pi\varphi^2_\bcond(z)(\varphi_\bcond(b)-\varphi_\bcond(z))}\right)^{\frac12}=f_{\Omega,\bcond}(z).
 $$
 
 It remains to justify the assumption that $H^{\bullet}$ are uniformly bounded away from $a$. Assume the contrary, i. e., that there is an $r_0>0$ such that $M_\delta=\max\limits_{\Od\cap \{z:|z-a|>r_0\}}H^\bullet\to\infty$. By a version of Harnack's principle for $H$ (see \cite[Proposition 3.11]{CHS2} or Appendix) the functions $M^{-1}_\delta H^\bullet$ are uniformly bounded on $\Od\cap \{z:|z-a|>r\}$ for \emph{all} $r>0$, and thus have subsequential limits. Any such limit $\tilde{h}$ is a non-negative harmonic function which is zero on $\partial \Omega\backslash \{a\}$ (since $M^{-1}_\delta H^{\circ,\bullet}\equiv M^{-1}_\delta\to 0$ on $(b^\delta c^\delta)$), but again by \cite[Remark 6.3]{CHS2} it has non-negative outer normal derivative on $(bc)$. Thus $\tilde{h}\equiv 0$ and 
 $$
 1=\max\limits_{\Od\cap \{z:|z-a|>r_0\}}M^{-1}_\delta H^\bullet\to 0,
 $$ 
 a contradiction which concludes the proof.
 \end{proof}

\section{Observables for general boundary conditions}
\label{sec: multipoint}

In this section we generalize the above construction and Proposition \ref{prop: convergence}. The boundary conditions $\bcond^\delta$ are now specified by three \emph{subsets} $\plus$, $\minus$ and $\free$ of $\partial\Od$, each containing an arbitrary number of arcs. We assume that there are in total $2k+m$ such arcs, of which $k$ carry the free boundary conditions. We denote the $2k$ endpoints of the free arcs by $b_1,\ldots,b_{2k}$ (ordered counterclockwise) so that the $i$-th free arc is $[b_{2i-1},b_{2i}]$. The boundary vertices separating the $\plus$ boundary arcs from the $\minus$ ones are denoted by $a_1,\dots,a_m$. As before, we assign the spin $+1$ (or $-1$ if we wish) to the faces of $\mesh \Z^2\setminus \Od$ adjacent to each free boundary arc. Let $a_{m+1},\dots,a_{m+s}\in\{b_1,\dots,b_{2k}\}$ be the boundary vertices that either separate a $\plus$ arc from a free arc with the assigned spin $-1$, or a $\minus$ arc from a free arc with the assigned spin $+1$. The collection $b_1,\dots,b_
{2k}$, $a_1,\dots,a_{m+s}$ of marked vertices determines the boundary conditions $\bcond$ uniquely up to global spin flip.

Denote by $\Conf(\Od,z_1,\dots,z_{2n})$, where $z_i$ are distinct vertices of the decorated graph, the set of all edge subsets $S$ such that all vertices  except for $z_i$ have an even degree in $S$. The low-temperature expansion
\begin{equation}
\label{eq: lt_multi}
Z(\Od,a_1,\dots,a_{m+s}):=\sum\limits_{S\in\Conf(\Od,a_1,\dots,a_{m+s})}x^{|S\backslash \free|} 
\end{equation}
endows  $\Conf(\Od,a_1,\dots,a_{m+s})$ with the probability measure equivalent to the Ising model on $\Od$ with the boundary conditions $\bcond^\delta$.

Denote by $\nor{b}_i$, $1\leq i\leq 2k$, the discrete outer normal corner at $b_i$ adjacent to the corresponding $\plus$ or $\minus$ arcs. Choose discrete outer normal \emph{edges} $\nor{a_i}$ at $a_i$ (assume for simplicity that such edges exist, although this is not essential, as explained in Remark~\ref{rem: a1}). Define the fermionic observable (using the idea of \cite{HThesis}) by
\begin{equation}
\label{eq: obs_many}
F(z)=F_{\Od,\bcond^\delta}(z)=i\eta_{\nor{a}_1}\sum\limits_{S\in \Conf(\Od,\nor{a}_1,\dots,\nor{a}_{m+s-1},z)}x^{|S\backslash \free|}e^{-i\wind(S)/2},
\end{equation}
where the winding factor of $S$ is defined as follows: we connect the points $\nor{a}_2,\dots,\nor{a}_{m+s-1}$ in pairs by $(m+s+2)/2$ arcs outside the domain, as shown on Fig. \ref{Fig: multi}. Then, every configuration $S$ can be decomposed into a collection of loops and a curve from $\nor{a}_1$ to $z$, without transversal intersections or self-intersections. We define $\wind(S)$ to be the winding of that curve. 

\begin{rem}
 In fact, $F_{\Od,\bcond^\delta}$ is a slight abuse of notation, since $F$ also depends on the choices made: the spins assigned to the free arcs, the order in which $a_i$ are listed, and the way they are  connected outside the domain.
\end{rem}

\begin{lemma}
 \label{lemma: shol_many}
 The observable $F$ defined by (\ref{eq: obs_many}) satisfies the s-holomorphicity condition (\ref{eq: s-hol}) whenever an edge $z$ and a corner $q$ are adjacent, provided that $z\neq \nor{a}_1,\dots,\nor{a}_{m+s-1}$ and $z\notin \free$. Its extension to $\free$ by $(\ref{eq: continuation})$ satisfies (\ref{eq: bc_bc}). If $z\neq \nor{a}_1,\dots,\nor{a}_{m+s-1}$ is a discrete outer normal at a vertex of $\plus\cup\minus$, then
 \begin{equation}
 \label{eq: bc_cb_many}
 F(z)=\eta_{z}Z(\Od,\nor{a}_1,\dots,\nor{a}_{m+s-1},z),
 \end{equation}
 where the sign of $\eta_z$ is defined by counterclockwise continuous extension from $\eta_{\nor{a}_1}$ along the boundary, multiplied by $-1$ for each of $\nor{a}_j$, $j=2,\dots,m+s-1$, encountered on the way.
\end{lemma}
\begin{proof}
 The proof is exactly the same as for Lemmas \ref{lemma: dhol}, \ref{lemma: bc}.
\end{proof}

\begin{lemma}
\label{lemma: Hmany}
 The function $H$ defined by applying (\ref{eq: defH}) to $F=F_{\Od,\bcond^\delta}$ satisfies the following properties. First, $H^{\circ}\equiv 0$ on the faces adjacent to $\plus\cup\minus$, and there exist constants $C_i=C_i(\Od,\bcond^\delta)\geq 0$, $1\leq i\leq k$, such that $H^{\bullet}\equiv C_i$ on the $i$-th free arc $[b_{2i-1},b_{2i}]$. With the extensions $H^\bullet\equiv 0$ to the vertices of $\delta \Z^2\setminus \Od$ adjacent $\plus\cup\minus$ and $H^\circ\equiv C_i$ to the faces of $\delta \Z^2\setminus \Od$ adjacent to $[b_{2i-1},b_{2i}]$, one has the discrete sub- and superharmonicity of $H^{\bullet, \circ}$ with respect to the modified Laplacian as in Lemma \ref{lemma: H}, namely, $\Delta H^\bullet(v)\geq 0$ provided that $v\notin \{a_1,\dots,a_m\}\cup\free$, and $\Delta H^\circ(u)\leq 0$ for all faces $u\in\Od$. Finally, $H^\bullet\geq 0$ on $\plus\cup\minus$ and $H^\circ\leq C_i$ on the faces of $\Od$ adjacent to $[b_{2i-1},b_{2i}]$. 
\end{lemma}

\begin{proof}
As in the proof of Lemma \ref{lemma: H}, we deduce from Lemma \ref{lemma: shol_many} that $H^{\bullet}\equiv \const$ on each free arc $[b_{2i-1},b_{2i}]$, and $H^{\circ}\equiv \const$ at the faces of $\delta\Z^2\setminus \Od$ adjacent to each arc complementary to $\free$. We must show that the latter constants are all the same, that is, that the absolute values of the jumps at the two ends of each $[b_{2i-1},b_{2i}]$ coincide. These values are equal to $\sqrt{2}\delta|F(\nor{b}_{2i-1})|^2$ and $\sqrt{2}\delta|F(\nor{b}_{2i})|^2$ respectively, and thus by (\ref{eq: bc_cb_many}) it is enough to prove that $Z(\Od,\nor{a}_1,\dots,\nor{a}_{m+s-1},\nor{b}_{2i})=Z(\Od,\nor{a}_1,\dots,\nor{a}_{m+s-1},\nor{b}_{2i-1})$. The weight-preserving bijection between the corresponding configuration sets given by taking the symmetric difference with $\nor{b}_{2i-1}\cup[b_{2i-1},b_{2i}]\cup\nor{b}_{2i}$ readily proves the identity. The remaining properties are proven as in Lemma \ref{lemma: H}.
\end{proof}

\begin{figure}[t]
\begin{center}
\includegraphics[width=0.8\textwidth]{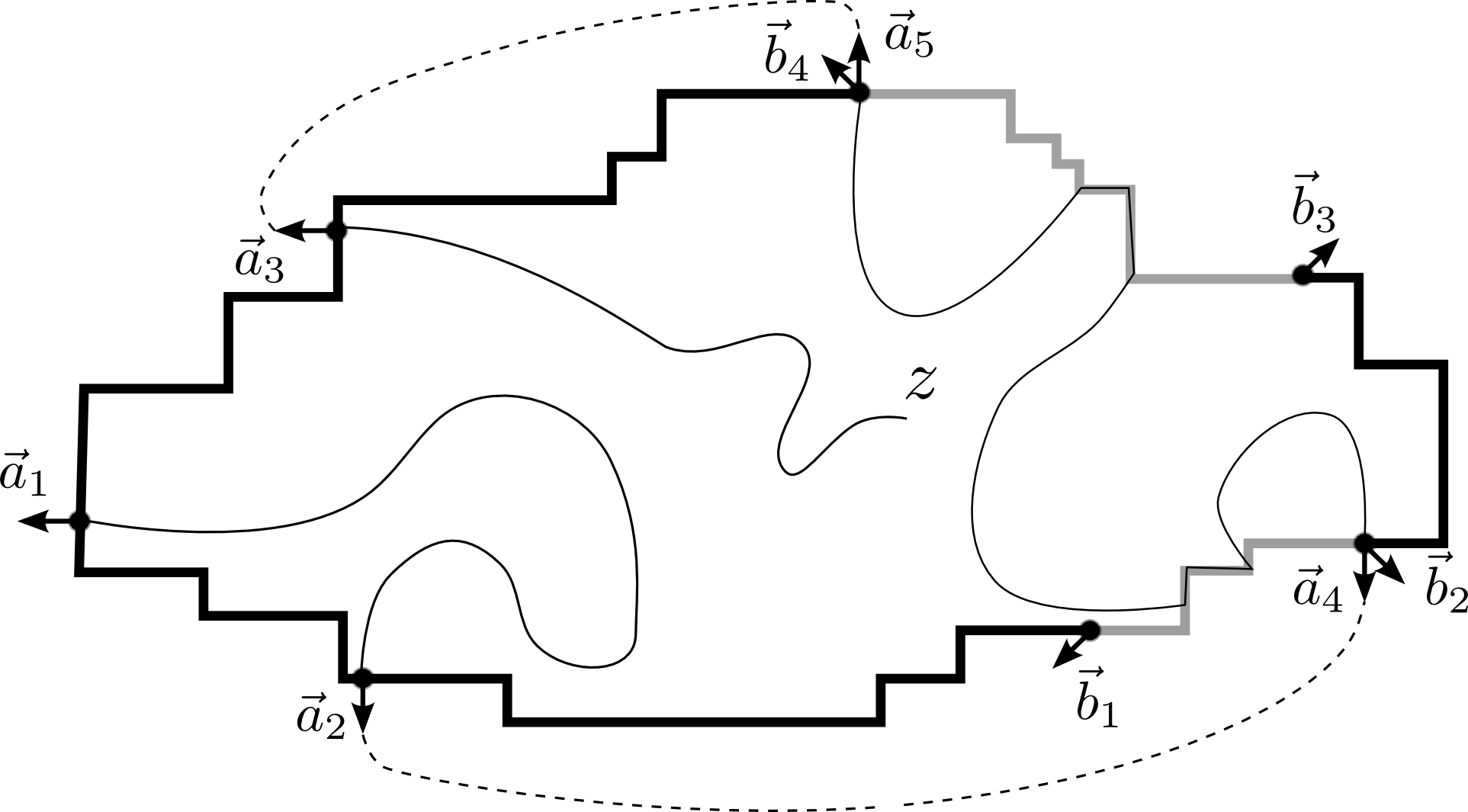}
\end{center}
\caption{\label{Fig: multi} Computation of the winding of a configuration $S$; loops not drawn: by adding the dashed lines, a collection of curves becomes a single curve from $\nor{a}$ to $z$. In this particular case, $w(S)=2\pi$.}
\end{figure}

We now define the functions $f_{\Omega,\bcond}$ which we will prove to be the scaling limits of~$F_{\Od,\bcond^\delta}$. Rather than giving a complicated explicit formula (see Proposition \ref{prop: linear}), we prefer a definition based on a boundary value problem and justified by Lemmas \ref{lemma: Hmany} and \ref{lemma: uniqueness}. From now on, we assume that $k>0$ (i. e. there is at least one free arc) and that $a_{m+s}=b_{2k}$. the latter condition does not lose generality: we can always ensure $b_{2k}\in\{a_{m+1},\ldots,a_{m+s}\}$ by choice of the spin assigned to the free arc $[b_{2k-1},b_{2k}]$. Given the boundary conditions $\bcond$ in the upper half-plane (that is, a collection of points $b_1<\dots<b_{2k}\in \R$ and $a_1,\dots,a_{m+s}\in \R$ with $a_{i}\notin [b_{2j-1};b_{2j}]$ for $1\leq i\leq m$, $1\leq j\leq k$ and $a_{i}\in \{b_{1},\ldots,b_{2k}\}$ for $m<i\leq m+s$), we define 
\begin{equation}
f_{\H,\mathcal{B}}:=\frac{P_\bcond(z)}{\prod\limits_{i=1}^k\sqrt{(z-b_{2i-1})(z-b_{2i})}\prod\limits_{i=1}^m(z-a_i)}. 
\label{eq: deff_many}
\end{equation}
where $P_\bcond(z)$ is the polynomial of degree $k+m-1$ with real coefficients such that $f_{\H,\mathcal{B}}$ satisfies the following conditions:
\begin{itemize}
 \item For all $1\leq i\leq m$,
 \begin{equation}
 \label{eq: residue}
 \lim\limits_{z\to a_i} \left(f_{\H,\bcond}(z)-\frac{\res_{a_i}f_{\H,\bcond}}{z-a_i}\right)=0;
\end{equation}
\item For all $1\leq i\leq k-1$,
\begin{equation}
\label{eq: freearcs}
 \lim\limits_{z\to b_{2i}}\sqrt{(z-b_{2i-1})(z-b_{2i})}f_{\H,\bcond}(z)=-\fsign_i\lim\limits_{z\to b_{2i-1}}\sqrt{(z-b_{2i-1})(z-b_{2i})}f_{\H,\bcond}(z); 
\end{equation}
where the sign $\fsign_i= \pm 1$ is equal to $(-1)^{|\{b_{2i-1},b_{2i}\}\cap \{a_{m+1},\dots,a_{m+s-1}\}|}$.
\item One has 
\begin{equation}
\label{eq: normalization}
\lim\limits_{z\to b_{2k}}\sqrt{\pi(z-b_{2k})}f_{\H,\bcond}(z)=1. 
\end{equation}
\end{itemize}

We will write $\bcond'\prec \bcond$ if boundary conditions $\bcond'$ have the same set $\free$ as $\bcond$ (i. e., $k'=k$ and $b'_1=b_1,\dots,b'_{2k}=b_{2k}$) and a smaller or equal set of points with $+/-$ changes (i. e., $m'\leq m$ and $\{a'_1,\dots,a'_{m'}\}\subseteq \{a_1,\dots,a_m\}$).
\begin{lemma}
\label{lemma: uniqueness}
 The polynomial $P_\bcond$ with the above properties exists and is unique. Moreover, if $h$ is a non-negative linear combination of the Poisson kernels at $a_i$, $1\leq i\leq m$, and the harmonic measures of $(b_{2i-1},b_{2i})$, $1\leq i \leq k$, such that the outer normal derivative of $h$ is non-negative on $\free$ and $\sqrt{2\partial_z h}$ is analytic in $\H$, then $\sqrt{2\partial_z h}=c f_{\H,\bcond'}$ for some boundary conditions $\bcond'\prec \bcond$ and some constant $c\in\R$.  
\end{lemma}
\begin{proof}
Define $h_\bcond(w)=\Im\int^wf^2_{\H,\bcond}(z)  dz$. Let us check that (\ref{eq: residue}) and (\ref{eq: freearcs}) imply that $h_\bcond$ satisfies all the properties of $h$ in the assertion. Since $f_{\H,\bcond}$ is real on $\plus\cup\minus$ and purely imaginary on $\free$, $h_\bcond$ is constant on each arc with the correct signs of the normal derivative. Note that (\ref{eq: residue}) can be written as
\begin{equation}
f_{\H,\bcond}(z)=\frac{\res_{a_i}f_{\H,\bcond}}{z-a_i}+O(z-a_i), \quad z\to a_i,
\label{eq: res_f} 
\end{equation}
and that $\res_{a_i}f_{\H,\bcond}$ is real; taking this to the square and integrating yields
\begin{equation}
h_{\bcond}(z)=-\Im\left[\frac{(\res_{a_i}f_{\H,\bcond})^2}{z-a_i}\right]+O(z-a_i),\quad z\to a_i.
\label{eq: res_h}
\end{equation}
In particular, there are no jumps in the constant values of $h_\bcond$ across $a_i$. Similarly, by~(\ref{eq: freearcs}), the jumps at $b_{2i}$ and $b_{2i-1}$ are negatives of each other for $1\leq i\leq k-1$. But $|\nabla h_{\bcond}(w)|=O(|w|^{-2})$ at infinity, so the net jump in the constant values of $h_\bcond$ along $\R$ must be zero. This shows that the jumps of $h_\bcond$ at $b_{2k}$ and $b_{2k-1}$ are also negatives of each other. 

Assume that $f_1$ and $f_2$ are both of the form (\ref{eq: deff_many}) and both satisfy (\ref{eq: residue}) -- (\ref{eq: normalization}) (with the same $\zeta_i$). Then $f_1-f_2$ satisfy (\ref{eq: residue}) -- (\ref{eq: freearcs}), hence $\tilde{h}(w):=\Im\int (f_1-f_2)^2$ is a non-negative harmonic function equal to zero on $\plus\cup\minus$. By (\ref{eq: normalization}), also $\tilde{h}\equiv 0$ on $(b_{2k-1},b_{2k})$. But its outer normal derivative is non-negative there, thus $\tilde{h}$ must vanish identically. This proves the uniqueness of $f_{\H,\mathcal{B}}$. To prove the existence, note that (\ref{eq: residue}) -- (\ref{eq: normalization}) gives $m+k$ linear equations on $m+k$ unknown coefficients of $P_\bcond$, and we have just proved that this system is non-degenerate. 

Conversely, if $h$ satisfies all the properties in the assertion, then  $f:=\sqrt{2\partial_z h}$  is real on $\plus\cup\minus$ and purely imaginary on $\free$. Also, $h$ obeys the expansions of the type~(\ref{eq: res_h}) at each $a_i$. Let the subset $\{a_1',\dots,a'_{m'}\}$ consist of the points of $\{a_1,\dots,a_m\}$ such that $\res_{a_i}f\neq 0$. For these points, (\ref{eq: res_h}) implies (\ref{eq: res_f}). Since the jumps of $h$ at the endpoints of each free arc are negatives of each other, $f$ obeys (\ref{eq: freearcs}) with some choice of signs $\zeta_i$. Similarly, since $h$ has a non-positive jump discontinuity at $b_{2k}$, the limit 
$$c:=\lim_{z\to b_{2k}}\sqrt{\pi(z-b_{2k})} f(z)$$
exists and is real. The above uniqueness argument shows that $f-cf_{\H,\bcond'}\equiv 0$, where $\bcond'$ is specified by the marked points $b_1'=b_1,\dots,b'_{2k}=b_{2k}$, $a'_1,\dots,a'_{m'}$ and the property that $f_{\H,\bcond'}$ obeys (\ref{eq: freearcs}) with the same signs as $f$.
\end{proof}

\begin{rem}
\label{rem: a1} The proof shows that $\res_{a_i}f_{\H,\bcond}\neq 0$ for $1\leq i\leq m$, and that the limits in (\ref{eq: freearcs}) are non-zero. Indeed, $\res_{a_i}f_{\H,\bcond}= 0$ would imply by (\ref{eq: residue}) that the normal derivative of $h_{\bcond}$ vanishes at $a_i$. Since $h_{\bcond}\geq 0$, this is impossible by the Harnack principle. If the limit in (\ref{eq: freearcs}) were zero, then $h_{\bcond}$ would be equal to zero and have non-negative outer normal derivative of $(b_{2i-1},b_{2i})$, which is impossible. 
\end{rem}

\begin{theorem}
\label{thm: multiple}
 Suppose the domains $(\Od,\bcond^\delta)$ approximate a simply connected domain $(\Omega,\bcond)$. Then
 $$
 \norF:=\frac{F_{\Od,\bcond^\delta}(\cdot)}{2^{\frac14}\sqrt{\delta}Z(\Od,\nor{a}_1,\dots,\nor{a}_{m+s-1},\nor{b}_{2k})}\longrightarrow f_{\Omega, \bcond}(\cdot):= (\varphi'(z))^{\frac12}f_{\H,\varphi(\bcond)}(\varphi(z))
$$
uniformly on compact subsets of $\Omega$, where $\varphi$ is any conformal map from $\Omega$ to $\H$.
\end{theorem}
\begin{proof}
 We may assume that for all $r>0$, the functions $H$ defined by (\ref{eq: defH}) from $\norF$ are uniformly bounded on $\Od\cap \{|z-a_i|>r,1\leq i\leq m\}$. This assumption can be justified a posteriori as in the proof of Proposition \ref{prop: convergence}. By \cite[Theorem 3.12]{CHS2} $\norF$ and $H$ have subsequential limits, say $f$ and $h=\Im\int f^2$ respectively. Using (\ref{eq: hm}) and \cite[Remark 6.3]{CHS2}, we see that the boundary properties of $H$ established in Lemma \ref{lemma: Hmany} survive in the limit, that is, $h_{\H}:=h\circ(\varphi^{-1})$ satisfies the conditions of Lemma \ref{lemma: uniqueness}, and hence $h_\H=\Im \int f^2_{\H,\bcond'}$, where $\bcond'\prec \bcond$. Our task is to show that actually $\bcond'=\bcond$; cf. excluding the case $\alpha=0$ in the proof of Proposition \ref{prop: convergence}.
 
 First, we show that $f_{\H,\bcond}$ and $f_{\H,\bcond'}$ must satisfy (\ref{eq: freearcs}) for $1\leq i\leq k$ with the same signs $\zeta_i$. Consider an arc $(b_{2i-1},b_{2i})$ and the auxiliary observable $\norF_{\bcond^\delta_i}$, where $\bcond_i$ stands for the following simple boundary conditions: free on $[b_{2i-1},b_{2i}]$ and $\plus$ elsewhere. (Hereinafter we drop $\Od$ from subscripts.) By Lemma~\ref{lemma: aa}, $\norF_{\bcond^\delta_i}$ converges to 
$$f_{\Omega,\bcond_i}(z)=(\varphi'(z))^\frac12\left(\frac{\varphi(b_{2i})-\varphi(b_{2i-1})}{\pi(\varphi(z)-\varphi(b_{2i-1}))(\varphi(z)-\varphi(b_{2i}))}\right)^{-\frac12}.$$ The relation (\ref{eq: bc_cb_many}), in particular the rule for the signs of $\eta_z$, implies the following: if $\{b_{2i-1},b_{2i}\}\cap\{a_{m+1},\dots,a_{m+s-1}\}$ contains one point (respectively, no or two points), 
 then the values of $\
|\norF_
{\bcond^\delta} + \norF_{\bcond^\delta_i}|$ at $\nor{b}_{2i-1}$ and $\nor{b}_{2i}$  are the same (respectively,  differ by 2$\min\{|\norF_
{\bcond^\delta}(\nor{b}_{2i-1})|,|\norF_{\bcond_i^\delta}(\nor{b}_{2i-1})|\}$). 
Taking into account Remark \ref{rem: a1}, we see that $\Im\int(f_{\H,\bcond'}+f_{\Omega,\bcond_i})^2=\lim_{\delta\to 0} \Im\int(\norF_{\bcond^\delta}+\norF_{\bcond^\delta_i})^2$ has jump 
discontinuities of the same size (respectively, of different sizes) at $b_{2i-1}$ and $b_{2i}$. It is easy to see that this condition fixes the sign in~(\ref{eq: freearcs}) in the correct way.

It remains to prove that $m'=m$, that is, that the singularities at $a_i$ do not disappear in the limit. Assume the contrary, and let $a_i\in\{a_1,\dots,a_m\}\setminus\{a'_1,\dots,a'_{m'}\}$. By induction on $m$, we know that that $\norF_{\bcond^{'\delta}}$ converges to $f_{\Omega,\bcond'}$, hence $F^{\dagger}=\norF_{\bcond^{\delta}}-\norF_{\bcond^{'\delta}}$ converges to zero uniformly on compact subsets of $\Omega$. Our goal is to deduce that for small $\delta$ there exist discrete outer normals $\nor{l},\nor{r}$ at some vertices of the two boundary arcs separated by $a_i$ for which $|F^{\dagger}(\nor{l})|<|\norF_{\bcond^{'\delta}}(\nor{l})|$ and $|F^{\dagger}(\nor{r})|<|\norF_{\bcond^{'\delta}}(\nor{r})|$. Taking into account that $F(\nor{z})\in l_{\nor{z}}$ for $\nor{z}$ a discrete outer normal, we see that these inequalities lead to the desired contradiction. Indeed, they imply that $\norF_{\bcond^{\delta}}(\nor{l})/\norF_{\bcond^{'\delta}}(\nor{l})$ and $\norF_{\bcond^{\delta}}(\nor{r})/\norF_{\bcond^{'
\delta}}(\nor{r})$ have the same the sign, whereas we know by (\ref{eq: bc_cb_many}) that these signs must be opposite due to the point $a_i$ between $\nor{l}$ and $\nor{r}$. 

Let $\gamma_1\subset \partial\Omega$ be a proper sub-arc of the $\plus$ or of the $\minus$ arc adjacent to $a_i$, $\gamma$ a proper sub-arc of $\gamma_1$, and $\gamma^\delta_1,\gamma^\delta$ approximations to $\gamma_1,\gamma$ in $\partial\Od$. Let $H^\dagger$ be constructed from $F^\dagger$, as usual, by integrating (\ref{eq: defH}) from some face of $\delta \Z^2\setminus \Od$ adjacent to $\gamma^\delta_1$. 
Then $H^\dagger\equiv 0$ at all such faces. Assume that $|F^{\dagger}(\nor{l})|\geq|\norF_{\bcond^{'\delta}}(\nor{l})|$ for all discrete outer normals $\nor{l}$ adjacent to $\gamma^\delta$.  
By (\ref{eq: defH}), this implies that $H^{\bullet\dagger}(v)\geq H_{\bcond^{'\delta}}^{\bullet}(v)$ for all the vertices $v$ of~$\gamma^\delta$. Denote $\Od_r:=\Od\cap\{|z-a_i|>r,1\leq i\leq m\}$ for a small fixed $r$. By (\ref{eq: hm}), one has $H^{\bullet\dagger}(v)\leq \text{hm}^\bullet_{\Od_r}(\pa\Od_r\setminus\gamma^\delta_1,v)\max_{\Od_r}H^{\bullet\dagger}$ and $H_{\bcond^{'\delta}}^{\bullet}(v)\geq H_{\bcond^{'\delta}}^{\circ}(u)\geq \text{hm}^\circ_{\Od}([b_{2k-1};b_{2k}],u)$, where $u$ is any face of $\Od$ incident to $v$, since $H_{\bcond^{'\delta}}^{\circ}\equiv 1$ on $[b_{2k-1},b_{2k}]$. Summing the resulting inequalities over the vertices of $\gamma^\delta$ yields
\begin{equation}
 \max_{\Od_r}H^{\bullet\dagger}\cdot\sum_{v\in\gamma^\delta}\text{hm}^\bullet_{\Od_r}(\pa\Od_r\setminus\gamma^\delta_1,v)\geq \sum_{u\sim\gamma^\delta}\text{hm}^\circ_{\Od}([b_{2k-1};b_{2k}],u), 
\label{eq: flux_contradiction}
 \end{equation}
where each face $u\sim\gamma^\delta$ is included into the last sum as many times as many of its adjacent edges belong to $\gamma^\delta$. 
By interpreting the sums as the flux of the gradient of the harmonic measure through $\gamma^\delta$ (see Appendix), it is not hard to see that there exist constants $C_{1,2}>0$ independent of $\delta$ such that 
\begin{eqnarray}
\sum_{v\in\gamma^\delta}\text{hm}^\bullet_{\Od_r}(\pa\Od_r\setminus\gamma^\delta_1,v)&\leq& C_1,\label{eq: flux1}\\
\sum_{f\sim\gamma^\delta}\text{hm}^\circ_{\Od}([b_{2k-1};b_{2k}],u)&\geq& C_2.\label{eq: flux2}
\end{eqnarray}
We claim that $\max_{\Od_r}H^{\bullet\dagger}\to 0$ as $\delta \to 0$. Indeed, as in the proof of Lemma \ref{lemma: H} we see that $H^{\dagger}$ is constant along each boundary arc and is bounded from below by the minimum of these constants (this time there might be jumps in these constants at $a_j\in \{a_1,\dots,a_m\}\setminus\{a'_1,\dots,a'_{m'}\} $). 
Convergence of $H^\dagger$ to zero in the bulk implies that all these constants must tend to zero, which in its turn implies that the maximum of $H^\dagger$ over $\Od_r$ tends to zero. 
Therefore, (\ref{eq: flux_contradiction}) fails for $\delta$ small enough, contradicting the assumption that $|F^{\dagger}(\nor{l})|\geq|\norF_{\bcond^{'\delta}}(\nor{l})|$ for all discrete outer normals at the vertices of $\gamma^\delta$. 
Applying the same argument to the other arc adjacent to $a_i$, we find the desired normal $\nor{r}$, thus concluding the proof.
\end{proof}

Theorem \ref{thm: multiple} is of certain independent interest because of the following corollary. Consider the critical FK-Ising model (that is, the random cluster model with $q=2$) on a domain $\Od$ with boundary conditions $\bcond^\delta$ specified by $2k$ marked boundary vertices $x_1,\dots,x_{2k}$ (listed counterclockwise), each of the boundary arcs $[x_1x_2],\dots,[x_{2k-1}x_{2k}]$ being wired and their complement left free. Denote by $\crossing(\Od,\bcond^\delta,i_1,\dots,i_r)$ the probability of the event that the arcs $[x_{2i_1-1}x_{2i_1}],\dots,[x_{2i_r-1}x_{2i_r}]$ belong to the same cluster. 
\begin{corollary}
Suppose the domains $(\Od,\bcond^\delta)$ approximate $(\Omega,\bcond)$. Then each of the probabilities $\crossing(\Od,\bcond^\delta,i_1,\dots,i_r)$ tends to a conformally invariant limit which is a quadratic irrational function of the images of $x_1,\dots,x_{2k}$ under a conformal map to $\H$.
\end{corollary}
\begin{proof}
Sample a random-cluster configuration in $\Od$, and assign a spin $+1$ to all vertices of the cluster attached to $[x_{2i_1-1},x_{2i_1}]$, and for each of the other clusters choose $\pm 1$ independently with probability $\frac12$. By the well-known Edwards-Sokal argument, the resulting spin configuration is distributed as in the critical Ising model with $+$ boundary conditions on the arc $[x_{2i_1-1},x_{2i_1}]$, free boundary conditions on $\partial\Od \backslash \cup_{i=1}^k[x_{2i-1},x_{2i}]$ and ``monochromatic'' boundary conditions on each $[x_{2i-1},x_{2i}]$ for $i\neq i_1$ (i. e., the spins on each arc are conditioned to be the same, but not fixed; we denote these random spins by $\sigma([x_{2i-1},x_{2i}])$). Given $\sigma_1=\pm 1,\dots, \sigma_k=\pm 1$, let $Z_{\sigma_1 \dots \sigma_k}$ denote the restriction of the partition function to the set of configurations in this model such that $\sigma([x_{2i-1},x_{2i}])=\sigma_i$. Then
\begin{multline}
\crossing(\Od,\bcond^\delta,i_1,\dots,i_r)= \E[\sigma([x_{2i_1-1},x_{2i_1}])\dots\sigma([x_{2i_r-1},x_{2i_r}])]\\= \frac{\sum_{\sigma\in\{\pm 1\}^k:\;\sigma_{i_1}\equiv 1}\sigma_{i_1}\dots\sigma_{i_r}Z_{\sigma_1 \dots \sigma_k}}{\sum_{\sigma\in\{\pm 1\}^k:\;\sigma_{i_1}\equiv 1}Z_{\sigma_1 \dots \sigma_k}}.
\label{eq: crossing} 
\end{multline}

Therefore, it suffices to prove the convergence of all the ratios $Z_{\sigma_1 \ldots \sigma_k}/Z_{+1 \ldots +1}$ to conformally invariant limits. To this end, consider the observable $F_{\bcond^\delta}$ for the boundary conditions $\bcond^\delta$ corresponding to $Z_{+1\ldots+1}$ (that is, take $b_{1}=x_2,\ldots,b_{2k-1}=x_{2k},b_{2k}:=x_1$, and  put $s=2$ and $a_1:=b_{2k-1}$, $a_2:=b_{2k}$) and apply (\ref{eq: bc_cb_many}) to $z=\nor{b}_{2k}$ and $z=\nor{b}_{2i}$. This yields
\begin{equation}
\left|\frac{F_{\bcond^\delta}(\nor{b}_{2i})}{F_{\bcond^\delta}(\nor{b}_{2k})}\right|=\frac{Z_{\sigma_1 \dots \sigma_k}}{Z_{+1\ldots+1}},  
\label{eq: F_to_Z} 
\end{equation}
where $\sigma_1=-1,\dots,\sigma_{i}=-1$, $\sigma_{i+1}=+1,\dots,\sigma_{k}=+1$. The left-hand side of (\ref{eq: F_to_Z}) is the jump at $b_{2i}$ in the boundary value of $H$, which by Theorem \ref{thm: multiple} tends to a conformally invariant quantity 
\begin{equation}
\label{eq: F_to_f}
\left|\frac{F_{\bcond^\delta}(\nor{b}_{2i})}{F_{\bcond^\delta}(\nor{b}_{2k})}\right|\stackrel{\delta\to 0}{\longrightarrow} \lim_{z\to\varphi(b_{2i})}\left|\sqrt{\pi(z-\varphi(b_{2i}))}f_{\H,\varphi(\bcond)}(z)\right|,
\end{equation}
which is a quadratic-irrational function of $\varphi(b_1),\dots,\varphi(b_{2k})$. 
The same reasoning applied to $F_{\bcond^\delta}$ with other boundary conditions $\bcond^\delta$ yields the corresponding result for all the ratios of the form $Z_{-\sigma_1\ldots-\sigma_i \sigma_{i+1}\ldots\sigma_{2k}}/Z_{\sigma_1\ldots\sigma_{2k}}$, and every ratio is a telescoping product thereof. 
\end{proof}
\begin{rem}
 The explicit expressions for the right-hand side of (\ref{eq: F_to_f}), and hence for the limits of $\crossing$, follows readily from Proposition \ref{prop: linear} below. Since those are rather complicated, we prefer not to write them down.
\end{rem}

\section{Convergence of interfaces}
\label{sec: interfaces}
Let $\gamma^\delta:=\{\nor{a}_1=\gamma^\delta_0,\gamma^\delta_1,\gamma^\delta_2\dots\}$ denote the random discrete interface starting at $a^\delta_1$ in the decomposition of $S\in\Conf(\Od,\bcond^\delta)$; we assume any deterministic or random rule to resolve ambiguities in the decomposition of $S$; for example, one may take the rightmost possible interface. Denote by $\gamma^\delta_{[0,n]}:=\{\gamma^\delta_0,\gamma^\delta_1,\dots,\gamma^\delta_n\}$ the initial segment of this interface containing $n+1$ edges. Let $\varphi^\delta(z)$ be a conformal map which maps $\Od$ to the upper half-plane $\H$ such that $\varphi(a_1^\delta)\neq \infty$.
Then $\varphi^\delta(\gamma^\delta)$ is a slit in $\H$ which can be described by Loewner's equation:
$$
 \partial_t g^\delta_t(z)=\frac{2}{g^\delta_t(z)-a_1^\delta(t)},\quad g_0^\delta(z)=z.
$$
Here $g^\delta_t$ is the conformal map from $\H^\delta_t$ to $\H$ satisfying hydrodynamic normalization $g^\delta_t(z)=z+o(1)$ at infinity, $\H^\delta_t$ is the unbounded connected component of $\H\backslash \varphi^\delta\left(\gamma^\delta_{[0,t]}\right)$, $t$~is the parametrization of $\gamma^\delta$ by twice the half-plane capacity of $\varphi^\delta\left(\gamma^\delta_{[0,t]}\right)$, and $a^\delta_1(t)\in\R$ is the random driving force. We denote by $\bcond_t^\delta$ the boundary conditions in $\H$ specified by the marked points $a^\delta_1(t),\ldots, a^\delta_m(t)$, $b^\delta_1(t),\ldots, b^\delta_{2k}(t)$, where  $a^\delta_i(t):= g^\delta_t(a^\delta_i(0))$ and  $a_i^\delta(0):=\varphi^\delta(a_i^\delta)$, $2\leq i\leq m$; similarly for $b_i^\delta(t)$, $1\leq i\leq 2k$.

Given boundary conditions $\bcond$ in $\H$, denote 
$$
 R(\bcond):=\res_{a_1} f_{\H,\bcond};
$$
note that $R$ is a quadratic irrational function of the points $a_1,\dots,a_{m},b_1,\dots,b_{2k}$ which specify $\bcond$; we give the explicit expression for $R(\bcond)$ in Proposition \ref{prop: res_explicit} below. Let $a_1(t)$ be the solution to the following system of stochastic differential equations:
\begin{equation}
\label{eq: drift}
\left\{
\begin{aligned}
da_1(t)=\sqrt{3}dB(t)-3\frac{\partial_{a_1}R(\bcond_t)}{R(\bcond_t)}dt,  \\
da_i(t)=\frac{2dt}{a_i(t)-a_1(t)},\quad 2\leq i\leq m;\\
db_i(t)=\frac{2dt}{b_i(t)-a_1(t)},\quad 1\leq i\leq 2k.
\end{aligned}
\right.
\end{equation}
where $B(t)$ is the standard Brownian motion in $\R$ and $\bcond_t$ stands for boundary conditions in $\H$ specified by $a_1(t),\dots,b_{2k}(t)$. Note that $a_i(t)=g_t(a_i(0))$, $2\leq i\leq m$ and $b_i(t)=g_t(b_i(0))$, $1\leq i\leq 2k$, where $g_t$ is the solution to Loewner's equation with the driving force~$a_1(t)$

If $\crs$ is a cross-cut in $\H$, let $T_\crs$ (respectively, $T_\crs^\delta$) denote the first time that the curve generated by $a_1(t)$ (respectively, by $a_1^\delta(t)$) hits $\crs$.
\begin{theorem}
\label{thm: interfaces} Suppose $(\Od,\bcond^\delta)$ approximate $(\Omega,\bcond)$ and the maps $\varphi^\delta$ are chosen so that they converge to $\varphi:\Omega\to\H$ with $\varphi(a_1)\neq \infty$ uniformly on compact subsets of $\Omega$. Then, for any cross-cut $\crs$ in $\H$ separating $\varphi(a_1)$ from other marked points and from infinity, the random processes $a^\delta_1(t)$ can be coupled to $a_1(t)$ (defined by (\ref{eq: drift}) with initial conditions $\bcond_0:=\varphi(\bcond^\delta)$) in such a way that $
 \sup_{[0,T^\delta_\crs\wedge T_\crs]}|a^\delta_1(t)-a_1(t)| 
 $ tends to zero in probability.
\end{theorem}
\begin{rem}
Since the tightness of interfaces is known for the critical Ising model (the ``no-six-arm estimate'' easily follows e.g. from \cite[Corollary 1.7]{CDH}), this theorem immediately implies convergence in the topology of curves. The results of \cite{AnttiStas}, although formally do not apply directly to our setup, also can be adapted to get the precompactness estimates and to simplify the proof below. The same a priori bounds allow one to extend the convergence up to the time interface hits other marked points or free boundary arcs.
\end{rem}

We will need several standard analytic facts. First, the set of all possible realizations of~$\H_t$ for $t<T_\crs$ (and of~$\H^\delta_t$ for $t<T^\delta_\crs$) is precompact with respect to the Carath\'eodory topology as seen from any point outside $\crs$; this in particular implies that all possible realizations of $\bcond^\delta_t$ and $\bcond_t$ belong to a compact subset of the set of  $(2k+m)$-tuples of \emph{distinct} points in $\R$. Second, if $g_t,g_{1,t}$ and $g_{2,t}$ are Loewner chains driven by Loewner parameters $a(t),a_1(t)$, $a_2(t)$ respectively, and their hulls at times $t, t_{1},t_{2}$ are separated from a compact set $\comp\subset \overline{\H}$ by a cross-cut $\crs$, then for all $z\in\comp$,
\begin{gather}
|g_{1,t}(z)-g_{2,t}(z)|\leq C \max\limits_{[0,t]}|a_1(t)-a_2(t)| \label{eq: lip}\\
\left|g_{t_2}(z)-g_{t_1}(z)-\frac{2(t_2-t_1)}{g_{t_1}(z)-a(t_1)}\right|\leq C|t_2-t_1|\left(\max_{t_1\leq t\leq t_2}|a(t)-a(t_1)|+|t_2-t_1|\right)\label{eq: regg}\\
 \left|g'_{t_2}(z)-g'_{t_1}(z)+\frac{2g'_{t_1}(z)(t_2-t_1)}{(g_{t_1}(z)-a(t_1))^2}\right|\leq C|t_2-t_1|\left(\max_{t_1\leq t\leq t_1}|a(t)-a(t_2)|+|t_2-t_1|\right)\label{eq: reggprime}
\end{gather}
with a constant $C$ depending only on $\mathcal{C}$ and $\crs$. See \cite{Zhan-LERW} or Appendix for proofs.

The proof of Theorem \ref{thm: interfaces} is based on the martingale property of the observable $\norF_{\Od,\bcond^\delta}$. Given $\gamma^\delta_{[0,n]}$, denote by $\Od(n)$ the connected component of 
$$
\Od \backslash \{\text{faces adjacent to the edges of } \gamma^\delta_{[0,n]}\}
$$ 
that contains all the marked points except for $a_1^\delta$, if such a component exists. Note that once $\gamma^\delta_{[0,n]}$ is known, one also knows the spins on its adjacent faces, thus $\Od(n)$ comes with natural boundary conditions which we denote by $\bcond^\delta(n)$. Denote by $\mathcal{F}_n$ the filtration generated by $\gamma^\delta_{[0,n]}$ and by $\mathcal{F}_t$ the one generated by $a^\delta_1|_{[0,t]}$; then $\mathcal{F}_t=\mathcal{F}_{n_t}$, where $n_t$ is the smallest $n$ such that $2\hcap(\varphi^\delta(\gamma^\delta_{[0,n]}))\geq t$. 
\begin{lemma}
\label{lemma: martingale}
For every edge $z\in\Od$, the process $\norF_{\Od(n),\bcond^\delta(n)}(z)$, stopped at the first $n$ such that either $z\notin \Od(n)$ or $\Od(n)$ seizes to exist, is a martingale with respect to $\mathcal{F}_n$. Given a cross-cut $\beta$ in $\H$ separating $\varphi(a_1)$ from infinity, other marked points and from a compact set $\mathcal{C}\subset \H$, one has  
  \begin{equation}
   \label{eq: approxMart}
  \E\left[\left.\sqrt{(g^\delta_{\tau_2})'(z)}f_{\H,\bcond_{\tau_2}^\delta}(g_{\tau_2}^\delta(z))-\sqrt{(g^\delta_{\tau_1})'}f_{\H,\bcond_{\tau_1}^\delta}(g_{\tau_1}^\delta(z))\right|\mathcal{F}^\delta_{\tau_1}\right]\stackrel{\delta\to 0}{\longrightarrow} 0
  \end{equation}
uniformly in $z\in\mathcal{C}$ and in $\mathcal{F}^\delta_t$-stopping times $\tau_1<\tau_2<T^\delta_\crs$.
\end{lemma}
\begin{proof}
 The first assertion follows from the fact that $F$ is normalized by the appropriate partition function; the proof repeats verbatim, e. g., one in \cite[Proposition 2.1]{Izy}. The second one follows from the first one, Theorem \ref{thm: multiple} and the compactness. Indeed, denote by $e_z$ the edge of $\Od$ closest to $\left(\varphi^\delta\right)^{-1}(z)$; then $e_z\in\Od(n_{\tau_2})$ for $\delta$ small enough.
 Clearly, $n_{\tau_{1,2}}$ are stopping times for $\mathcal{F}_n$. Applying the optional stopping theorem to the martingale $F_{\Od(n),\bcond^\delta(n)}(e_z)-F_{\Od(n_{\tau_1}),\bcond^\delta(n_{\tau_1})}(e_z)$, we infer that its expected value at $n_{\tau_2}$ given $\mathcal{F}_{\tau_1}$ is zero. The quantity $|((\varphi^\delta)^{-1})'|$ is bounded from above and below uniformly in $\delta$ and over the compact set $\comp$. Hence, it suffices to show that 
 \begin{equation}
 \label{eq: uniform}
 F_{\Od(n_t), \bcond(n_t)}(e_z)-\sqrt{\frac{(g^\delta_{t})'(z)}{((\varphi^\delta)^{-1})'(z)}}f_{\H,\bcond^\delta_t}((g_{t}^\delta(z)) \stackrel{\delta\to 0}{\longrightarrow} 0 
 \end{equation}
uniformly over $z\in \mathcal{C}$ and over all $t$ and all possible $\gamma^\delta_{[0,t]}$ such that $t<T^\delta_\crs$, i. e. that (\ref{eq: uniform}) holds for every sequence $\delta^{(i)}$, $z^{(i)}\in\comp$, $t^{(i)}$, $\gamma^{\delta^{(i)}}_{[0,t^{(i)}]}$ of such data. The set of all possible realizations of $\Od\backslash \gamma^\delta_{[0,t]}$ is Carath\'eodory compact as seen from any point of $\varphi^{-1}(\comp)$, hence  we may assume that $\Omega^{\delta^{(i)}}\backslash \gamma^{\delta{(i)}}_{[0,t^{(i)}]}\stackrel{\text{Cara}}{\longrightarrow}\Omega'$. Then also $\Omega^{\delta^{(i)}}(n_{t^{(i)}})\stackrel{\text{Cara}}{\longrightarrow}\Omega'$, since $\Od(n_{t})\subset \Od\backslash \gamma^\delta_{[0,t]}$ and any point of their difference can be separated from $(\varphi)^{-1}(\mathcal{C})$ by a cross-cut of length at most  $2\delta$. By one more extraction, we may assume that $\gamma_{t^{(i)}}^{\delta^{(i)}}$ also converges to some prime end $a'\in\pa\Omega'$ and that $z^{(i)}\to z$. But then (\
ref{eq: uniform}) follows directly from Theorem~\ref{thm: multiple}.
\end{proof}

\begin{lemma}
\label{lemma: LSW}
If two $\mathcal{F}^\delta_t$-stopping times $\tau_1<\tau_2<T^\delta_\crs$ almost surely satisfy the inequalities $\tau_2-\tau_1\leq\epsilon^2$, $\max\limits_{\tau_1\leq t\leq \tau_2}|a^\delta_1(t)-a^\delta_1(\tau_1)|\leq\epsilon$, then
 \begin{eqnarray}
  |\E[\Delta_a+3\frac{\partial_{a_1}R(\bcond^\delta_{\tau_1})}{R(\bcond^\delta_{\tau_1})}\Delta_{\tau}|\mathcal{F}_{\tau_1}]|<C(\bcond^\delta_0,\crs)\epsilon^3\\
  |\E[\Delta_a^2-3\Delta_{\tau}|\mathcal{F}_{\tau_1}]|<C(\bcond^\delta_0,\crs)\epsilon^3,
 \end{eqnarray}
 provided that $\delta<\delta_0(\bcond^\delta_0,\crs,\epsilon)$, where $\Delta_{\tau}:=\tau_2-\tau_1$ and $\Delta_a:=a^\delta_1(\tau_2)-a^\delta_1(\tau_1)$
\end{lemma}
\begin{proof}
In the proof below, the constants in $O(\cdot)$ may depend on $\bcond_0$ and $\crs$. The idea is to expand (\ref{eq: approxMart}) in the small parameters $\Delta_{\tau}$, $\Delta_a$ up to the order $\epsilon^3$. Choose a compact set $\comp\subset \H$ with a non-empty interior separated by $\crs$ from $a^\delta_1$. By (\ref{eq: regg}), (\ref{eq: reggprime}), with the notation $u(z):=g^\delta_{\tau_1}(z)-a^\delta_1(\tau_1)$, one has 
$$
g_{\tau_2}^\delta(z)-g_{\tau_1}^\delta(z)= \frac{2}{u(z)}\Delta_\tau+O(\epsilon^3);
$$
$$
(g^\delta_{\tau_2})'(z)^\frac12-(g^\delta_{\tau_1})'(z)^\frac12=-\Delta_\tau\frac{(g^\delta_{\tau_1})'(z)^\frac12}{u^2(z)}+O(\epsilon^3).
$$
whenever $z\in\comp\cup\{a^\delta_2(0),\ldots, b^\delta_{2k}(0)\}$.  By (\ref{eq: residue}), we can write $f_{\H,\bcond}(z)=\frac{R(\bcond)}{z-a_1}+\dob(\bcond, z)$, where $\dob(\bcond, z)$ is analytic at $a_1$ and $\dob(\bcond, a_1)=0$. Therefore one has, for $z\in\comp$,
 $$
 (g^\delta_{\tau_2})'(z)^\frac12f_{\H,\bcond_{\tau_2}^\delta}(g_{\tau_2}^\delta(z))-(g^\delta_{\tau_1})'(z)^\frac12f_{\H,\bcond_{\tau_1}^\delta}(g_{\tau_1}^\delta(z))=
 $$
 \begin{equation}
  (g^\delta_{\tau_1})'(z)^\frac12\left(\frac{Q_3}{u^3(z)}+\frac{Q_2}{u^2(z)}+\frac{Q_1}{u(z)}+Q_0(u(z))\right)+O(\epsilon^3),
  \label{eq: expand}
 \end{equation}
 where
 \begin{eqnarray}
  Q_3:=(\Delta_a^2-3\Delta_t)R(\bcond_{\tau_1}^\delta)\label{eq: q3}\\ 
  Q_2:=\Delta^2_a\partial_{a_1}R(\bcond_{\tau_1}^\delta)+\Delta_aR(\bcond_{\tau_1}^\delta). \label{eq: q2}
 \end{eqnarray}
The explicit form of $Q_1$ and $Q_0$ is not important, it suffices to note that they are polynomials in $\Delta_a,\Delta_\tau$, $R(\bcond_{\tau_1}^\delta)$, $\dob(\bcond^\delta_{\tau_1}, g_{\tau_1}^\delta(z))$ and first and second derivatives of the last two quantities with respect to the parameters; $Q_1$ does not depend on $z$. 

By Lemma \ref{lemma: martingale}, there exists $\delta_0>0$ such that for all $\delta<\delta_0$ and $z\in \comp$, the expectation of the left-hand side of (\ref{eq: expand}) given $\mathcal{F}^\delta_{\tau_1}$ is less than $\epsilon^3$.  By compactness, $|(g^\delta_{t})'(z)|$ is uniformly bounded from below over $z\in \comp$ and $t<T^\delta_{\crs}$. Therefore, for $\delta<\delta_0$, one has 
\begin{equation}
\frac{Q^*_3}{u^3(z)}+\frac{Q^*_2}{u^2(z)}+\frac{Q^*_1}{u(z)}+Q^*_0(u(z))=O(\epsilon^3) 
\label{eq: o_eps_3}
\end{equation}
uniformly over $z\in\comp$, where $Q^*_i:=\E[Q_i|\mathcal{F}^\delta_{\tau_1}]$, $i=0,1,2,3$. We claim that this identity implies $\max\{Q^*_1,Q^*_2,Q^*_3\}=O(\epsilon^3)$. Indeed, assume the contrary. Then there exist sequences of hulls $K^{(n)}$ in $\H$, points $a^{(n)}_1\in \R$, numbers $\epsilon^{n}$, $Q^{(n)}_i$ and quadratic irrational functions $Q^{(n)}_0$ with the following properties: \textbf{(i)} $K^{(n)}$ and $g^{-1}_{K^{(n)}}(a^{(n)}_1)$ are separated by $\crs$ from infinity and other marked points; \textbf{(ii)} for each $n$, the function $u\mapsto Q^{(n)}_0(u)$ belongs to the finite-dimensional space spanned by $\dob(g_{K^{(n)}}(\bcond^\delta_0), u + a^{(n)}_1)$ and its first and second partial derivatives with respect to the positions of the marked points; \textbf{(iii)} (\ref{eq: o_eps_3}) holds true with $Q^*_i=Q^{(n)}_i$, $\epsilon=\epsilon^{(n)}$ and $u(z)=g_{K^{(n)}}(z)-a^{(n)}_1$; \textbf{(iv)} $M^{(n)}(\epsilon^{(n)})^{-3}$ tends to infinity, where $M^{(n)}:=\max\{Q^{(n)}_1,
Q^{(n)}_2,Q^{(n)}_3\}$. By compactness, we may  assume that $K^{(n)}$, $a^{(n)}_1$, $(M^{(n)})^{-1}Q^{(n)}_i$, $i=1,2,3,$ converge to $K^{*}$, $a^{*}$ and $q_i$ respectively, with at least one of $q_i$ non-zero. Thus (\ref{eq: o_eps_3}) implies that $(M^{(n)})^{-1}Q^{(n)}_0(u)$ converges to $-q_3/u^3-q_2/u^2-q_1/u$ uniformly in $u\in g_{K^*}(\comp)$. On the other hand, the limit of $(M^{(n)})^{-1}Q^{(n)}_0$ must belong to the space spanned by $\dob(g_{K^{*}}(\bcond^\delta_0), u + a^{*}_1)$ and its first and second partial derivatives and thus it must be analytic at the origin. This contradiction proves the claim.
 
 By Remark \ref{rem: a1}, $R(\bcond_{t}^\delta)$ does not vanish, and thus by the compactness, it is uniformly bounded away from $0$ over $t<T^\delta_\crs$. Thus the lemma follows readily from (\ref{eq: q3}) -- (\ref{eq: q2}) and  the equations $\E[Q_{2,3}|\mathcal{F}^\delta_{\tau_1}]=O(\epsilon^3)$.
 \end{proof}
\begin{proof}[Proof of Theorem \ref{thm: interfaces}]
 Fix a cross-cut $\crs_1$ that separates $a_1(0)$ and $\crs$ from infinity and other marked points. Define recursively the stopping times $0=\tau_1,\tau_2,\dots\tau_{N(\crs_1)}$, $\tau_i$ being the maximal time so that $\tau_{i-1}$ and $\tau_i$ satisfy the conditions of Lemma \ref{lemma: LSW} with the cross-cut $\crs_1$, and define the process $w(t)$ by
$$
w(\tau_i):=a^\delta_1(\tau_i)+\sum\limits_{j=0}^{i-1}3\frac{\partial_{a_1}R(\bcond^\delta_{\tau_{j}})}{R(\bcond^\delta_{\tau_{j}})}(\tau_{j+1}-\tau_{j}).
$$
extending linearly between $\tau_j$. Recall that all possible realization of $\bcond^\delta_{t}$, $t<T^\delta_{\crs_1}$, form a precompact subset of the set of $(2k+m)$-tuples of distinct points of $\R$; thus $R$ is uniformly bounded from below and $\pa_{a_1}R/R$ is uniformly Lipschitz in its parameters over this subset. From this and from the bounds $\tau_{i+1}-\tau_{i}\leq \epsilon^2$, $\max\limits_{\tau_i\leq t\leq \tau_{i+1}}|a^\delta_1(t)-a^\delta_1(\tau_i)|\leq\epsilon$, we infer 
\begin{equation}
\label{eq: w_int}
w(t):=a^\delta_1(t)+\int_0^t\frac{\partial_{a_1}R(\bcond^\delta_{s})}{R(\bcond^\delta_{s})}ds+O(\epsilon),
\end{equation}
provided that $t\leq\tau_{N(\crs_1)}$. Continue the process $w(t)$ beyond the hitting time of $\crs_1$ by taking, for $i\geq N(\crs_1)$, $\tau_{i+1}:=\tau_i+\epsilon^2$; $w(\tau_{i+1}):= w(\tau_{i})\pm\sqrt{3}\epsilon$, with signs chosen at random with probability $1/2$ and independently. By Lemma \ref{lemma: LSW}, for all $i>0$, $\E[w_{i}-w_{i-1}|\mathcal{F}_{\tau_{i-1}}]=O(\epsilon^3)$ and $\E[(w_{i}-w_{i-1})^2-3(\tau_i-\tau_{i-1})|\mathcal{F}_{\tau_{i-1}}]=O(\epsilon^3)$, hence the standard argument (see \cite{LSW}) shows that there is a coupling of $w(t)$ to the standard Brownian motion $\sqrt{3}B(t)$ such that $\P[\max_{0<t<T}|w(t)-\sqrt{3}B(t)|<\sqrt{\epsilon}]>1-\sqrt{\epsilon}$, where $T=4\hcap(\crs_1)$. If $a_1(t)$ denotes the strong solution to (\ref{eq: drift}) with this $B_t$, then on this event, for all $t<T_{\crs_1}\wedge T^\delta_{\crs_1}$,
\begin{align*}
 |a_1^\delta(t)-a_1(t)|\leq |w(t)-\sqrt{3}B(t)|+3\int_0^t\left|\frac{\partial_{a_1}R(\bcond^\delta_{s})}{R(\bcond^\delta_{s})}-\frac{\partial_{a_1}R(\bcond_{s})}{R(\bcond_{s})}\right|ds+O(\epsilon)\\ \leq C(\crs_1)\int_0^t\max_{[0,s]}|a_1^\delta-a_1|ds+O(\sqrt{\epsilon}),
\end{align*}
where the first inequality follows from (\ref{eq: drift})  and (\ref{eq: w_int}) and the second one from (\ref{eq: lip}) and the uniform Lipschitzness of $\partial_{a_1}R/R$. The theorem now follows by Gronwall's lemma.
\end{proof}

\section{Explicit expressions for observables}
\label{sec: corollaries}

The main goal of this section is to give explicit expressions for the observables $f_{\H,\bcond}$ and hence for the drift terms $3\partial_{a_1}\log R(\bcond)$ and for the crossing probabilities $\crossing(\H,\bcond,i_1,\dots,i_r)$. We only do the cases $m=0$ and $m=1$; as we note in Remark \ref{rem: big_m} below, the formulae for arbitrary $m$ can be obtained from those by a straightforward limiting procedure. Let us introduce some notation. Denote by $\chi_{ij}$ the cross-ratio of the endpoints of $i$-th and $j$-th free boundary arc:
$$
\chi_{ij}=\chi_{ji}=[b_{2i-1};b_{2j-1};b_{2i};b_{2j}]=\frac{(b_{2i-1}-b_{2j-1})(b_{2i}-b_{2j})}{(b_{2i-1}-b_{2j})(b_{2i}-b_{2j-1})},
$$
and, given an auxiliary point $x\in \R$, put
$$
\psi_{ij}(x)=\frac{(x-b_{2j-1})(b_{2i-1}-b_{2j})}{(x-b_{2j})(b_{2i-1}-b_{2j-1})},
$$
so that the multiplication by $\psi_{ij}(x)$ acts by substituting $x$ into $\chi_{ij}$ in the place of $b_{2i-1}$ :
$$
[x;b_{2j-1};b_{2i};b_{2j}]=\chi_{ij}\psi_{ij}(x).
$$

\begin{proposition}
\label{prop: linear}
Given boundary conditions $\bcond$ in $\H$ with $m=0$, the observable $f_{\H,\bcond}$ reads
 $$
 f_{\H,\bcond}=C\cdot\prod_{1\leq i\leq k}\left(\frac{z-b_{2i-1}}{z-b_{2i}}\right)^{\frac12}\left(\sum_{1\leq i\leq k}\frac{p_i}{z-b_{2i-1}}\right),
 $$
 where $C$ does not depend on $z$, and 
\begin{gather}
p_r=\frac{b_{2r}-b_{2r-1}}{b_{2r}-b_{2k-1}}
\sum_{\substack{s\in\{\pm 1\}^{k}\\s_k\equiv-s_r\equiv 1}}
\prod_{\substack{1\leq i\leq k-1\\ \fsign_i=-1}}s_i 
\prod_{\substack{i< j}}\chi_{ij}^\frac{s_is_j}{4}\prod _{\substack{i\neq k,r}} \psi_{ri}(b_{2k-1})^{\frac{1-s_i}{2}},\quad 1\leq r\leq k-1,\label{eq: p_r}\\ \label{eq: p_k}
p_k=  \sum_{\substack{s\in\{\pm 1\}^{k}\\s_k\equiv 1}}
\prod_{\substack{1\leq i\leq k-1\\ \fsign_i=-1}}s_i 
\prod_{\substack{i< j}}\chi_{ij}^\frac{s_is_j}{4}.
\end{gather}
\end{proposition}
\begin{proof}
We look for a solution to the linear system (\ref{eq: freearcs}) -- (\ref{eq: normalization}) (note that (\ref{eq: residue}) is vacuous since $m=0$) by writing the unknown polynomial $P_\bcond$ is the form 
$$
P_\bcond(z)=\sum_i p_i\prod_{j\neq i}(z-b_{2j-1}).$$ We temporarily replace the last equation (\ref{eq: normalization}) by the equation $p_k=1$ (this will only change the overall normalization). Then it is straightforward to check that in the unknowns $p_i$, (\ref{eq: freearcs}) becomes a system with the right-hand side 
$$
\bar{v}:=\left(\frac{1}{b_{2k-1}-b_2},\frac{1}{b_{2k-1}-b_4},\dots,\frac{1}{b_{2k-1}-b_{2k-2}}\right)^T
$$
and with the matrix $A=D+C$, where $D$ is the diagonal matrix with entries 
$$
 D_{ii}=\frac{\fsign_i}{b_{2i}-b_{2i-1}}\prod\limits_{1\leq j\neq i\leq k}\chi^{\frac12}_{ij},\quad 1\leq i\leq k-1
$$
 and $C$ is the Cauchy matrix:
$$
C_{ij}= \frac{1}{b_{2i}-b_{2j-1}}.
$$
To solve the system by Cramer's rule, we have to compute $\det A$ and $\det A_{[r]}$ (where $A_{[r]}$ is the matrix $A$ with $r$-th column replaced by $\bar{v}$). Recall that given a subset $S$ of indices, the principal minor $\det(C_S)$ of the Cauchy matrix $C$ is given by 
$$
\det(C_S)=\prod\limits_{\substack{i\neq j\\ i,j\in S}}\chi^{\frac12}_{ij}\prod\limits_{i\in S}\frac{1}{b_{2i}-b_{2i-1}}.
$$
Therefore
\begin{multline*}
\det A = \sum\limits_{S\subset\{1,\dots,k-1\}} \det D_S \det C_{\bar{S}}\\=
\prod_{1\leq i<k-1}\frac{1}{b_{2i}-b_{2i-1}}
\sum_{\substack{S\subset\{1,\dots,k\}\\k\notin S}}
\prod_{\substack{i \in S}}\fsign_i
\prod_{\substack{i \in S\\j\neq i}}\chi_{ij}^{\frac12}
\prod_{\substack{i,j\notin S\\i,j\neq k}}\chi_{ij}^{\frac12}=\\
\prod_{1\leq i<k-1}\frac{1}{b_{2i}-b_{2i-1}}
\prod_{\substack{j \neq i\\j,i\neq k}}\chi^{\frac38}_{ij}
\prod_{\substack{i\neq k}}\chi_{ki}^{\frac{1}{4}}
\sum_{\substack{s\in\{\pm 1\}^{k}\\s_k=1}}
\prod_{\substack{1\leq i\leq k-1\\ \fsign_i=-1}}s_i
\prod_{j\neq i}\chi_{ij}^{\frac{s_is_j}{8}}.
\end{multline*}
The computation of $\det A_{[r]}$ is similar; the only difference is that $D_{rr}$ gets replaced by zero (which restricts the sum to $S$ such that $r\notin S$) and that in the Cauchy matrix, $b_{2r-1}$ gets replaced by $b_{2k-1}$ (which introduces the factor $\frac{b_{2r}-b_{2r-1}}{b_{2r}-b_{2k-1}}$ in front and the factors $\psi_{rj}(b_{2k-1})$ inside the sum). Canceling the common factors and changing the normalization once again, we arrive at the assertion.
\end{proof}

\begin{proposition}
\label{prop: res_explicit}
If $\bcond$ are boundary conditions in $\H$ such that $m=1$, then  
\begin{multline}
R(\bcond)=C(a_1-b_{2k})\prod_{1\leq i\leq k}\left(\frac{a_1-b_{2i-1}}{a_1-b_{2i}}\right)^{\frac12}\times\\\times\left(
\sum_{\substack{s\in\{\pm 1\}^{k}\\s_k\equiv-1}}
\prod_{\substack{1\leq i\leq k\\ \fsign_i=-1}}s_i 
\prod_{\substack{1\leq i<j\leq k}}\chi_{ij}^\frac{s_is_j}{4}\prod _{\substack{1\leq i\leq k-1}} \psi_{ki}(a_1)^{\frac{1-s_i}{2}}\right)^{-1} 
\label{eq: res_explicit}
\end{multline}
where the factor $C$ does not depend on $a_1$ (and hence is unimportant for the drift term in Theorem \ref{thm: interfaces}).
\end{proposition}
\begin{proof}
We will reduce this proposition to the computation in the previous one. First, it follows from the proof of Lemma \ref{lemma: uniqueness} that (\ref{eq: residue}) -- (\ref{eq: normalization}) imply (\ref{eq: freearcs}) for $i=k$ with some $\zeta_k$ (since the jumps of $h_{\bcond}$ at $b_{2k-1}$ and $b_{2k}$ must be the same). Actually, by Remark~\ref{rem: a1}, the function $f_{\H, \bcond}$ changes its sign at each $a_i$ and across each free arc with $\zeta_i=-1$. Since the total number of sign changes must be odd (taking into account the behavior at infinity), we have $\zeta_{k}=-\zeta_1\ldots\zeta_{k-1}(-1)^m$, and since $m+s$ is even, this is consistent with the definition of $\zeta$'s after (\ref{eq: freearcs}). Conversely, if $m=1$ and a non-zero function $f$ of the form (\ref{eq: deff_many}) satisfies (\ref{eq: freearcs}) for $1\leq i\leq k$, the same arguments imply that there must be a sign change at $a_1$ and at the same time that $h=\Im\int f^2$ has no jump in the constant boundary conditions at $a_
1$; 
therefore (\ref{eq: residue}) holds.

Summarizing, the equation (\ref{eq: residue}) can be replaced with (\ref{eq: freearcs}) for $i=k$, and this yields an equivalent system. Now look for $P_\bcond$ in the form 
$$P_\bcond(z)=\sum_i p_i(z-a_1)\prod_{j\neq i}(z-b_{2j-1})+p_{k+1}\prod(z-b_{2j-1}).$$ After replacing temporarily the normalization~(\ref{eq: normalization}) with $p_{k+1}=1$, we get exactly the same system as in Proposition \ref{prop: linear} with $k$ replaced by $k+1$ and $b_{2k+1}=b_{2k+2}=a_1$ (thus $\chi_{k+1,i}=1$, $1\leq i\leq k$). Since we know that the system is non-degenerate, the determinant of $A$ is non-zero, and so the sum in (\ref{eq: p_k}) is also non-zero. Hence, up to a non-zero multiplicative constant, the solution is given by (\ref{eq: p_r}) -- (\ref{eq: p_k}) with the above-mentioned substitutions. Changing the normalization to (\ref{eq: normalization}), picking up the residue at $a_1$ and dropping the factors independent on $a_1$ concludes the proof.
\end{proof}

\begin{rem}
\label{rem: big_m}
Starting with the expression (\ref{eq: res_explicit}), one could pick some other $i$ with $\zeta_i=-1$ and let $b_{2i-1}$ tend to $b_{2i}$. This way one obtains explicit formulae for $m>1$. When taking the limit, the coefficient $C$ in front might in general tend to zero and the second factor to infinity, hence one has to use the l'Hospital rule; the resulting expressions gets complicated and we prefer not to write them down. 
\end{rem}

To illustrate (\ref{eq: res_explicit}), we specialize it to a few cases with a small number of marked points. In the case $m=1$ and $k=1$ (that is, the $+/-/\text{free}$ boundary conditions), the sum contains just one term and all the products are empty. Hence
$$
R=C(a_1-b_1)^\frac12(a_1-b_2)^\frac12,
$$
and plugging this into (\ref{eq: drift}) yields SLE${}_3(-\frac32, -\frac32)$, recovering the result of \cite{HonKyt}. If $m=1$ and $k=2$ (which is the $+/-/\text{free}/+/\text{free}$ boundary conditions), choose $\fsign_1=\fsign_2=1$; then there are two terms in the sum, and 
\begin{multline*}
R=C \cdot\frac{(a_1-b_1)^\frac12(a_1-b_3)^\frac12(a_1-b_4)^\frac12}{(a_1-b_2)^\frac12\left(\chi_{12}^{-\frac14}-\chi_{12}^{\frac14}\psi_{21}(a_1)\right)}=\\=C_1\cdot\frac{(a_1-b_1)^\frac12(a_1-b_3)^\frac12(a_1-b_4)^\frac12(a_1-b_2)^\frac12}{(a_1-b_2)\left(\frac{(b_1-b_4)(b_1-b_3)}{(b_2-b_4)(b_2-b_3)}\right)^{\frac12}+(a_1-b_1)}. 
\end{multline*}
Taking $b_2$ to infinity, one gets the 5-point formula in the introduction, and merging $b_3$ with $b_4$ yields the 4-point one. 

We conclude by a corollary concerning the probability of a \emph{spin-crossing event}, i. e. the probability that in a domain with four marked points on the boundary there is a nearest-neighbor chain of pluses or minuses (diagonal jumps allowed) connecting two opposite sides. 
\begin{corollary}
Suppose $(\Od,\bcond^\delta)$ approximate $(\Omega,\bcond)$, where $\bcond$ stands for $+/-/+/\;\text{free}$ boundary conditions. Then the probability that there is a $+$ crossing between $(b^\delta_2 a^\delta_2)$ and $(a^\delta_1 b^\delta_1)$ (respectively, $-$ crossing between $(a^\delta_2 a^\delta_1)$ and $(b^\delta_1 b^\delta_2)$) tends to $1-G_{+/-/+/\text{free}}(\lambda)$ (respectively, $G_{+/-/+/\text{free}}(\lambda)$), where
$$
 G_{+/-/+/\text{free}}(\lambda)=\left(\int\limits_0^1 \frac{s^\frac23ds}{(1-s)^{\frac13}(2-s)^{2}}\right)^{-1}\int\limits_0^\lambda \frac{s^\frac23ds}{(1-s)^{\frac13}(2-s)^{2}},
$$
$\lambda =\frac{\varphi(a_1)-\varphi(a_2)}{\varphi(b_1)-\varphi(a_2)}$ and $\varphi$ is a conformal map from $\Omega$ to $\H$ with $\varphi(b_2)=\infty$.  
\end{corollary}
\begin{rem}
 The corresponding probabilities for other boundary conditions are given by 
\begin{eqnarray}
 G_{+/-/+/-}(\lambda)\label{eq: pmpm}=\left(\int\limits_0^1 \frac{s^\frac23(1-s)^{\frac23}ds}{(1-s+s^2)}\right)^{-1}\int\limits_0^\lambda \frac{s^\frac23(1-s)^{\frac23}ds}{(1-s+s^2)}\\
G_{+/-/\text{free}/\text{free}}(\lambda)=\left(\int\limits_0^1 \frac{ds}{s^{\frac13}(1-s)^{\frac13}}\right)^{-1}\int\limits_0^\lambda \frac{ds}{s^{\frac13}(1-s)^{\frac13}}\label{eq: pmff}
\end{eqnarray}
The formula (\ref{eq: pmpm}) was conjectured in \cite{BBK05} and proved in \cite{Izy_disser}, and (\ref{eq: pmff}) follows from the result of \cite{HonKyt} and a calculation in \cite{BB}. We do not know explicit formulae for other boundary conditions. In \cite{BDH}, the conformal invariance is proven for the $\text{free}/\text{free}/\text{free}/\text{free}$ and $+/\text{free}/\text{free}/\text{free}$ cases. Combining the techniques of \cite{BDH} with the computation of drifts in 5-point geometry, it should be possible to extend this result to the remaining $+/\text{free}/+/\text{free}$ case.
\end{rem}
\begin{proof}
 It is straightforward to check that $G(\lambda(t)):=G\left(\frac{a_1(t)-a_2(t)}{b_1(t)-a_2(t)}\right)$ is a martingale for the process (\ref{eq: drift}) with 
 $$
 R=\frac{\sqrt{a_1-b_1}(a_1-a_2)}{a_1+a_2-2b_1},
 $$
 and since $G(0)=0$, $G(1)=1$, we deduce that $G(\lambda(0))$ is the probability that the Loewner chain driven by this process swallows $b_1(0)$ before $a_2(0)$.  Given this, it is not hard to deduce from Theorem \ref{thm: interfaces} and \cite[Corollary 1.7]{CDH} the following corollary: the probability that the interface in $\Od$ starting at $a_1^\delta$ hits $(b_1^\delta b_2^\delta)$ before $a_2^\delta$ tends to $G(\lambda(0))$ as $\delta\to 0$ (see \cite[Proposition 5.26]{Izy_disser} for a detailed argument). This event (for the rightmost possible interface in the decomposition of $S$) is exactly the ``$-$'' crossing event; considering the complementary event for the leftmost possible interface gives the ``$+$'' crossing probability.
\end{proof}

\section{Appendix}
\begin{proof}[Proof of Lemma \ref{lemma: dhol}]
 Consider the bijection $p:S\mapsto S\bigtriangleup (zq)$ from $\Conf(\Od,\nor{a},q)$ to $\Conf(\Od,\nor{a},z)$, where $(zq)$ is the shortest graph path from the midpoint of $z$ to $q$, and $\bigtriangleup$ denotes the symmetric difference. It suffices to check that for any $S\in\Conf(\Od,\nor{a},q)$
 $$x^{|S\backslash\free|}e^{-i\frac{\wind(S)}{2}}=\text{Pr}_{l_q}x^{|p(S)\backslash\free|}e^{-i\frac{\wind(p(S))}{2}}.$$ Taking into account that the contribution of a half-edge and of a corner to $x^{|S\backslash\free|}$ equals $\sqrt{x}$ and $\sqrt{x}\cos\frac\pi8$ respectively, and using that $z \notin \free$, we find
 \begin{eqnarray*}
 x^{|p(S)\backslash\free|}e^{-i\wind(p(S))/2}=x^{|S\backslash\free|}e^{-i\frac{\wind(S)}{2}}e^{\pm \frac{i\pi}{8}}\sqrt{x}\left(\sqrt{x}\cos\frac{\pi}8\right)^{-1} \quad \text{if }z\in S \\
 x^{|p(S)\backslash\free|}e^{-i\wind(p(S))/2}=x^{|S\backslash\free|}e^{-i\frac{\wind(S)}{2}}e^{\pm \frac{3i\pi}{8}}\left(x\cos\frac{\pi}8\right)^{-1} \quad \text{if }z\notin S. \\
 \end{eqnarray*}
 Since $e^{-i\frac{\wind(S)}{2}}\in l_q$, the lemma follows from the elementary identity $x=\sqrt{2}-1=\frac{\cos 3\pi/8}{\cos \pi/8}$.
 \end{proof}

 \begin{proof}[Proof of the second clause of Lemma \ref{lemma: H}]
  We express $\Delta H^{\circ}$ at a face $u$ via $\tilde{F}(c_j)$, where $c_j$, $j=1,2,3,4$, are the corners adjacent to $u$ (listed counterclockwise starting from the lower right). Let $z$ be the edge that separates $u$ from $u+\delta$. If $z\notin \free$, then by (\ref{eq: defH}) and (\ref{eq: orthogonal}) one has
 $$
\frac{1}{\sqrt{2}\delta}(H^\circ(u+\delta)-H^\circ(u))=|\tilde{F}(c_1)|^2+|\tilde{F}(c_2)|^2-|\tilde{F}(e)|^2 
 $$
 which by (\ref{eq: continuation}) equals $-|\tilde{F}(c_1)|^2-|\tilde{F}(c_2)|^2-2i (\tilde{F}(c_1)\overline{\tilde{F}(c_2)} -\tilde{F}(c_2)\overline{\tilde{F}(c_1)})$. If $z \in \free$, then, with the extension of $H^\circ$ we have made, $H^\circ(u+\delta)=1=H^\bullet(u+\frac{\delta-i}{2})$. Hence by~(\ref{eq: defH}), one has
 $$
\frac{2\sqrt{2}-2}{\sqrt{2}\delta}(H^\circ(u+\delta)-H^\circ(u))=(2\sqrt{2}-2)|\tilde{F}(c_1)|^2, 
 $$
which by (\ref{eq: bc_proof}) also equals $-|\tilde{F}(c_1)|^2-|\tilde{F}(c_2)|^2-2i (\tilde{F}(c_1)\overline{\tilde{F}(c_2)} -\tilde{F}(c_2)\overline{\tilde{F}(c_1)}$ (recall that $2\sqrt{2}-2$ is the weight with which that difference contributes to the modified Laplacian). Summing similar identities for the four edges adjacent to $u$ yields, with the notation $c_5:=c_1$,
\begin{eqnarray*}
(\sqrt{2}\delta)^{-1}\Delta H^\circ(u)=-2\sum\limits_{j=1,2,3,4} |\tilde{F}(c_j)|^2 -2\sum\limits_{j=1,2,3,4} (i\tilde{F}(c_j) \overline{\tilde{F}(c_{j+1}})+\tilde{F}(c_{j+1}) \overline{i\tilde{F}(c_{j}}))= \\
-2|\tilde{F}(c_1)+i\tilde{F}(c_2)-\tilde{F}(c_3)-i\tilde{F}(c_4)|^2+4\Re(\tilde{F}(c_1)\overline{\tilde{F}(c_3)})+4\Re(\tilde{F}(c_2)\overline{\tilde{F}(c_4)}),
\end{eqnarray*}
and the last two terms are equal to zero since $\tilde{F}(c_j)\in l_{c_j}$ and $l_{c_{1,2}}$ are orthogonal, respectively, to $l_{c_{3,4}}$. The proof for $\Delta H^\bullet$ is similar.
 \end{proof}
 
 \begin{proof}[Proof of Harnack's principle for $H$]
  Suppose $0<r<r_0$, and let $v_0$ be the vertex of $\Od\backslash \{z:|z-a|>r\}$ such that  $H^\bullet(v_0)=\max\limits_{\Od\cap \{z:|z-a|>r\}}H^\bullet$. By subharmonicity of $H^\bullet$, we can construct a path $v_0\sim v_1\sim\dots$ of neighboring vertices such that $H^\bullet(v_{j+1})\geq H^\bullet(v_{j})$ which can only end at $a^\delta$. If $\gamma$ denotes the set of faces of $\Od$ adjacent to this path, then $H^\circ|_{\gamma}\geq \const H^\bullet(v_0)$ with some absolute positive constant (see \cite[Remark 3.10]{CHS2}). Hence by (\ref{eq: hm}), 
for any face $u$ of $\Od$, $H^\circ(u)\geq \const H^\bullet(v_0)\text{hm}^\circ(\gamma,u)$. If for every $\delta$ we choose $u$ to be a face of $\Od$ close to some fixed point $z\in\Omega\cap \{z:|z-a|>r_0+\epsilon\}$, then it is easy to see that $\text{hm}^\circ(\gamma,u)$ is bounded from below by a positive constant (depending only on $\Omega$, $r$ and $z$, but not on $\delta$). Hence $$\max\limits_{\Od\cap \{z:|z-a|>r\}}H=H^\bullet(v_0)\leq C(\Omega,r,r_0)\cdot H^\circ(u)\leq C(\Omega,r,r_0)\cdot \max\limits_{\Od\cap \{z:|z-a|>r_0\}}H.$$
 \end{proof}
\begin{proof}[Proof of (\ref{eq: flux1}) and (\ref{eq: flux2}).]
 We will actually prove that the sum in (\ref{eq: flux2}) converges as $\delta\to 0$ to the flux of $\nabla\text{hm}_{\Omega}([b_{2k-1};b_{2k}],\cdot)$ through $\gamma$, which is a positive conformal invariant (given by $-\int_\gamma\pa_n\text{hm}_{\Omega}([b_{2k-1};b_{2k}],z)|dz|$ whenever $\pa\Omega$ is smooth). Assume that the endpoints $\gamma_l,\gamma_r$ of $\gamma$ can be connected by a broken line $\beta$ in $\Omega$ in such a way that $\dist(z,\pa \Omega)\geq c\min\{|z-\gamma_l|;|z-\gamma_r|\}$, $z\in \beta$. Since any boundary arc $\gamma'$ can be approximated by an arc $\gamma$ with this property (choose two points $z_l$, $z_r\in \Omega$ close to the prime ends $\gamma'_l, \gamma'_r$, and let $\gamma$ be the arc bounded by the points of $\pa \Omega$ closest to $z_l$ and $z_r$), and due to the fact that $-\int_\gamma\pa_n\text{hm}_{\Omega}([b_{2k-1};b_{2k}],\cdot)$ is monotone and continuous in the endpoints of $\gamma$, this does not lose any generality.
 
 Fix $\epsilon>0$ and let $z^\epsilon_{l,r}\in\beta$ be such that $|z^\epsilon_{l,r}-\gamma_{l,r}|=\epsilon$. Let $\beta_\epsilon$ be the part of $\beta$ between $z^\epsilon_{l}$ and $z^\epsilon_{r}$. Define $\beta^\delta_\epsilon$ to be the broken line that consists of $\beta_\epsilon$ and the two segments connecting $z^\epsilon_{l,r}$ to the nearest points in $\pa\Od$. We will assume that the boundary arc $\gamma^\delta$ is bounded by these nearest points; by monotonicity the result can be extended to arbitrary approximations of $\gamma$. Denote for shortness $\text{hm}^\circ(\cdot):=\text{hm}^\circ_{\Od}([b_{2k-1};b_{2k}],\cdot)$. By summing the identity $\Delta \text{hm}^\circ(u) = 0$ over the faces $u$ enclosed by the union of $\gamma^\delta$ and $\beta^\delta$, we infer 
 $$
 \sum_{u\sim\gamma^\delta}\text{hm}^\circ(u)=\sum_{e\cap \beta^\delta_\epsilon\neq \emptyset}\text{hm}^\circ(e_{out})-\text{hm}^\circ(e_{in}),
 $$
 where the last sum is over the edges $e=(e_{out},e_{in})$ that intersect $\beta^\delta_\epsilon$. Since the gradients of the discrete harmonic measures converge uniformly on compact subsets of $\Omega$, the part of the sum in the right-hand side corresponding to $\beta_\epsilon$ converges to the flux of $\nabla\text{hm}_{\Omega}([b_{2k-1};b_{2k}],\cdot)$ through $\beta_\epsilon$. We claim that the contribution of the two remaining segments is bounded from above by $C\cdot \epsilon^\sigma$, where $C$ and $\sigma>0$ are absolute constants. Indeed, denote $R_e=\dist(e,\pa\Od)$, where $e$ is an edge on one of these two segments. The weak discrete Beurling estimate (see, e.g., \cite{CHS1}) shows that $\max_{B_{R_e/2}(e)}\text{hm}^\circ\leq C_1\cdot R_e^\sigma$, therefore the discrete Harnack inequality yields $|\text{hm}^\circ(e_{in})-\text{hm}^\circ(e_{out})|\leq C_2\cdot \delta R_e^{\sigma-1}$, and summing this estimate over the two segments proves the claim. Taking $\epsilon$ to zero concludes the 
proof of convergence of the sum in (\ref{eq: flux2}).
 
 To prove (\ref{eq: flux1}), the only additional work is to take into account the multiplicities. Write $\text{hm}^\bullet(\cdot):=\text{hm}^\bullet_{\Od_r}(\pa\Od_r\setminus\gamma^\delta_1,\cdot)$ and denote by $d(v)$, $v\in\gamma^\delta$, the number of neighbors of $v$ not in $\Od$, with the convention introduced in Remark \ref{rem: at_a}.  Note that every $v\in\gamma^\delta$ such that $d(v)=0$ has a neighbor $v'\in\gamma^\delta$  with $d(v')\geq 1$. By discrete harmonicity and positivity, one has $\text{hm}^\bullet(v)\leq 4\text{hm}^\bullet(v')$. Therefore, 
$$
\sum_{v\in\gamma^\delta}\text{hm}^\bullet(v)\leq 10\sum_{v\in\gamma^\delta}d(v)\text{hm}^\bullet(v),
$$
and the last quantity converges to the flux of $\nabla\text{hm}_{\Omega_r}(\pa\Omega_r\setminus\gamma_1,\cdot)$ through $\gamma$. 
 \end{proof}
 
 \begin{proof}[Proof of (\ref{eq: lip})--(\ref{eq: reggprime})] It follows from Loewner's equation that
 $$|\partial_t g_{1,t}(z)-\partial_t g_{2,t}(z)|\leq2\frac{|g_{1,t}(z)-g_{2,t}(z)|+|a_1(t)-a_2(t)|}{|g_{1,t}(z)-a_1(t)||g_{2,t}(z)-a_2(t)|}.$$ By compactness, the denominator is bounded away from zero by a constant depending on $\crs$ and $\comp$ only. Now Gronwall's lemma readily implies
 \begin{equation}
 \label{eq: Gronwall}
 |g_{1,t}(z)-g_{2,t}(z)|\leq C_1\max_{[0,t]}|a_1-a_2|(e^{C_1t}-1)\leq C_1\max_{[0,t]}|a_1-a_2|(e^{2C_1\hcap(\crs)}-1),
 \end{equation} which is (\ref{eq: lip}). 
 Note that it suffices to prove (\ref{eq: regg}) for $t_2-t_1$ small enough (since the left-hand side of (\ref{eq: regg}) is bounded by compactness). Put, for $t_1\leq t\leq t_2$, $g_{1,t}:= g_t$ and $g_{2,t}:= \sqrt{(g_{t_1}(z)-a(t_1))^2+4(t_2-t_1)}+a(t_1)$. These are two Loewner chains on the interval $[t_1,t_2]$ with the same initial data $g_{t_1}(z)$, driven by $a_1(t)=a(t)$ and $a_2(t)\equiv a(t_1)$ respectively; the latter chain generates a vertical slit. If a cross-cut $\crs_1$ separates $\crs$ from $\comp$, then by compactness, $\min|g_{t_1}(\crs_1)-a(t_1)|>c(\crs,\crs_1)>0$; hence the vertical slit stays inside $g_{t_1}(\crs_1)$ for $t\leq t_1+c_1(\crs,\crs_1)$. The same argument as for (\ref{eq: Gronwall}) now yields
 $$
 |g_{1,t_2}(z)-g_{2,t_2}(z)| \leq C_2\max_{[t_1,t_2]}|a(t)-a(t_1)|(e^{C_2(t_2-t_1)}-1) \leq C_3\max_{[t_1,t_2]}|a(t)-a(t_1)||t_2-t_1|.
 $$
 From the explicit formula for $g_{2,t_2}(z)$, bearing in mind that $|g_{t_1}(z)-a(t_1)|$ is uniformly bounded away from zero for $z\in\comp$, we get 
$$
 \left|g_{2,t_2}(z)-g_{t_1}(z)-\frac{2(t_2-t_1)}{g_{t_1}(z)-a(t_1)}\right| \leq C_4 |t_2-t_1|^2,
 $$ 
and (\ref{eq: regg}) follows. The proof of (\ref{eq: reggprime}) is similar: differentiating Loewner's equation, we get, for two Loewner chains,
 $$
 |\partial_t g'_{1,t}(z)-\partial_t g'_{2,t}(z)|\leq \frac{2|g'_{1,t}(z)-g'_{2,t}(z)|}{(g_{1,t}(z)-a_1(t))^2}+\left|\frac{2|g'_{2,t}(z)|}{(g_{1,t}(z)-a_1(t))^2}-\frac{2|g'_{2,t}(z)|}{(g_{2,t}(z)-a_2(t))^2}\right|.
 $$
Since the denominators are bounded from below and $|g'_{1,t}(z)|$ from above, this yields
$$|\partial_t g'_{1,t}(z)-\partial_t g'_{2,t}(z)|\leq C_6 |g'_{1,t}(z)-g'_{2,t}(z)|+C_7|g_{1,t}(z)-g_{2,t}(z)|+C_8|a_1(t)-a_2(t)|,$$
and hence by (\ref{eq: Gronwall}) and Gronwall's lemma, 
$$|g'_{1,t}(z)-g'_{2,t}(z)|\leq C_9 \max_{[0,t]}|a_1-a_2|(e^{C_{10}t}-1).$$ 
To conclude the proof, apply this, as above, to $a_1(t)=a(t)$ and $a_2(t)\equiv a(t_1)$.
\end{proof}

\end{document}